\documentclass{llncs}

\usepackage{graphicx}
\usepackage{amssymb}
\usepackage{amsfonts}
\usepackage{amsmath}
\usepackage{color}
\usepackage{a4wide}
\usepackage{bibentry}
\usepackage{bm}

\DeclareMathOperator{\CH}{CH}

\providecommand{\pb}[1]{{\sc #1} problem}

 \newtheorem{observation}{Observation}
\newtheorem{open}{Open problem}

\newcommand{\rodrigo}[2][says]{** \textsc{Rodrigo #1:} \textcolor{red}{\textsl{#2}} **}

\newcommand{\mati}[2][says]{*** \textsc{matias #1:} \textcolor{blue}{\textsl{#2}} ***}

\definecolor{darkgreen}{rgb}{0,0.4,0}
\newcommand{\alex}[2][says]{*** \textsc{alex #1:} \textcolor{darkgreen}{\textsl{#2}} ***}

\let\doendproof\endproof
\renewcommand\endproof{~\hfill\qed\doendproof}

\newcommand{\ShoLong}[2]{#2} %New command to have short and long versions coexist. Change the "#1" to "#2" for a long version, and vice versa

\newcommand{\etalem}{}
\title{New results on stabbing segments with a polygon}

 \author{Jos\'e Miguel D\'iaz-B\'a\~nez\inst{1}\fnmsep\thanks{Partially supported by COFLA-Excellence Projects P09-TIC-4840 and P12-TIC-1362 (Junta de Andaluc\'{i}a).}\fnmsep$^\ddagger$
\and Matias Korman\inst{2,3}\fnmsep\thanks{With the support of the Secretary for Universities and Research of the Ministry of Economy and Knowledge of the Government of Catalonia and the European Union.}\fnmsep$^\ddagger$
\and Pablo P\'erez-Lantero\inst{4}\fnmsep\thanks{Partially supported by grant FONDECYT 11110069 and
 MEC MTM2009-08652.}
\and\\ Alexander Pilz\inst{5}\fnmsep\thanks{Recipient of a DOC-fellowship of the Austrian Academy of Sciences at the Institute for Software Technology, Graz University of Technology, Austria.}
\and Carlos Seara\inst{6}\fnmsep\thanks{Partially supported by the ESF EUROCORES programme EuroGIGA -ComPoSe IP04-MICINN Project EUI-EURC-2011-4306.}
\and Rodrigo I. Silveira\inst{7,6}\fnmsep\thanks{Funded by Portuguese funds through the  Center for Research and Development in Mathematics and Applications (CIDMA) and by the Portuguese Foundation for Science and Technology (FCT), within project PEst-OE/MAT/UI4106/2014, and by FCT grant SFRH/BPD/88455/2012.}\fnmsep$^\ddagger$
}

\institute{
Dept. Matem\'atica Aplicada II, Universidad de
Sevilla, Spain.% \texttt{dbanez@us.es}.
\and  National Institute of Informatics, Tokyo, Japan.
\and JST, ERATO, Kawarabayashi Large Graph Project.
\and Escuela de Ingenier\'ia Civil en Inform\'atica, Universidad de Valpara\'iso, Chile. %\texttt{pablo.perez@uv.cl}.
\and Institute for Software Technology, Graz University of Technology, Austria.
\and Dept.\ Matem\`atica Aplicada II, Universitat Polit\`ecnica de Catalunya, Spain. 
\and Dept. de  Matem\'atica \& CIDMA, Universidade de Aveiro, Portugal. %\texttt{\{matias.korman | carlos.seara | rodrigo.silveira\}@upc.edu}.
}
\graphicspath{{figures/}}

\usepackage{lineno}
\begin{document}

\maketitle

\begin{abstract}
We consider a natural variation of the concept of \emph{stabbing} a set of segments with a simple polygon:
a segment $s$ is stabbed by a simple polygon~$\mathcal{P}$ if at least 
one endpoint of $s$ is contained in~$\mathcal{P}$, and a segment set $S$ is stabbed by $\mathcal{P}$ if 
$\mathcal{P}$ stabs every element of $S$. 
Given a segment set $S$, we study the problem of finding a simple polygon~$\mathcal{P}$ stabbing $S$
in a way that some measure of $\mathcal{P}$ (such as area or perimeter) is optimized.
We show that if the elements of $S$ are pairwise disjoint, the problem can be solved in polynomial time. 
In particular, this solves an open problem posed by L\"offler and van Kreveld~[Algorithmica 56(2), 236--269 (2010)] about finding a maximum perimeter convex hull for a set of imprecise points modeled as line segments. 
Our algorithm can also be extended to work for a more general problem, in which instead of segments, the
set $S$ consists of a collection of point sets with pairwise disjoint convex hulls.
We also prove that for general segments our stabbing problem is NP-hard.
\end{abstract}

\section{Introduction}\label{intro}

Let $S$ be a set of $n$ straight line segments (\emph{segments} for
short) in the plane. The problem of stabbing $S$ with different
types of stabbers (in the computer science literature) or
transversals (in the mathematics literature) has been widely studied during the last two decades.

Rappaport~\cite{R} considered the case in which the stabber is a simple polygon.
Specifically, he studied the following problem: a simple polygon $\mathcal{P}$
is a \emph{polygon transversal} of $S$ if we have
$\mathcal{P}\cap s\not= \emptyset$ for all $s\in S$; that is, every
segment in $S$ has at least one point in $\mathcal{P}$. A simple polygon $\mathcal{P}$
is a \emph{minimum polygon transversal} of $S$ if $\mathcal{P}$ is a polygon
transversal of $S$ and all other transversal polygons have equal or larger perimeter. Rappaport observed that such a polygon always exists, is convex,
and may not be unique. He gave an $O(3^m n + n\log n)$ time algorithm for computing one, 
where $m$ is the number of different
segment directions. Several approximation algorithms are
known~\cite{DM,HR}, but determining if the general problem can be solved
in polynomial time is still an intriguing open problem.

Arkin {\etalem et al.}~\cite{ADKMPSY} considered a similar problem: $S$ is
\emph{stabbable} if there exists a convex polygon whose boundary
$\mathcal{C}$ intersects every segment in $S$; the closed convex
chain $\mathcal{C}$ is then called a (convex) \emph{transversal}
or \emph{stabber} of $S$. Note that in this variant there is not always a solution.
Arkin {\etalem et al.}~\cite{ADKMPSY} proved that deciding
whether $S$ is stabbable is NP-hard.
%
%\mati{to do: check this claim}

In this paper we also consider the problem of stabbing the set $S$
by a simple polygon, but with a different criterion that is
between the two criteria above.
More concretely, we use the following definition:

\begin{definition}
A segment $s\in S$ is stabbed by a simple polygon $\mathcal{P}$ if \emph{at least} one of
the two endpoints of $s$ is contained in $\mathcal{P}$.
The set $S$ is stabbed by $\mathcal{P}$ if every segment of $S$ is stabbed by
$\mathcal{P}$.
\end{definition}

With this definition we study the {\sc Stabbing Polygon Problem (SPP)}, defined as finding a simple polygon $\mathcal{P}$ that stabs a given set $S$ of segments and optimizes some objective function. The main focus of this paper is the case in which we want to minimize the perimeter of the stabber (denoted by {\sc MinPerSPP}). Naturally, one could also study the maximization variants of the problem, or even the case in which we measure the area instead. However, with the current formulation these problems are trivial, since there exists stabbers of arbitrarily large area/perimeter (or arbitrarily small area, realized, e.g., by a simple polygon resulting from ``thickening'' a plane tree spanned by segment endpoints by an arbitrarily small amount).
Instead, we follow the formulation of L\"{o}ffler and van Kreveld~\cite{kl-alchip-08,LK} and formulate the {\sc MinAreaSPP}, {\sc MaxPerSPP} and {\sc MaxAreaSPP} as follows: given a set $S$ of segments, select one endpoint of each segment such that the convex hull of the selected endpoints has minimum area, maximum perimeter, or maximum area, respectively.
It is straightforward to verify that any polygon obtained in this way is a stabber.
Also, the optimal solution of the {\sc MinPerSPP} can be obtained with this approach (as the convex hull of any stabber is also a stabber with at most the same perimeter; see also \cite[Lemma~1]{R}).

Note that the four variants that we consider are discrete (that is, we do not consider the interior of the segments).
One could alter the definition of the {\sc SPP} by saying that the input is a collection of \emph{pairs of points} instead of segments. However, as we will show later, the segments play an important role in establishing the difficulty of the problem, hence we keep the segment formulation. 

Although the differences between the {\sc MinPerSPP} and the problems studied by Rappaport~\cite{R} and Arkin {\etalem et al.}~\cite{ADKMPSY} may look small, we observe that the problems are substantially different. The difference with the problem studied by Rappaport~\cite{R} is that $\mathcal{P}$ can be a stabber and have both endpoints of a segment of $s\in S$ outside $\mathcal{P}$ (provided that the interior of $s$ is stabbed by $\mathcal{P}$), whereas we force one of the endpoints to be in $\mathcal{P}$. One of the common properties of both problems is that the optimal solution is a convex polygon and that it always exists (the convex hull of $S$ is always a stabbing polygon according to both definitions). On the other hand, a solution of an instance of the {\sc MinPerSPP}  may fully contain a segment of $S$. This is not allowed in the stabber definition used by Arkin {\etalem et al.}~\cite{ADKMPSY}. Thus, we can say that our problem is between the two mentioned ones.

%\vspace{0.2cm}

\subsection{Related work}
Prior to the paper by Rappaport~\cite{R}, Meijer and Rappaport~\cite{MR} solved the problem of computing a minimum perimeter polygon transversal for a set of $n$ parallel segments in optimal $\Theta(n\log n)$ time.
%Javad et al.~\cite{mohades} described an $O(n^2 \log n)$ time algorithm for disjoint segments, which, however, turned out to be incorrect.\footnote{Personal communication with A. Javad.}
Mukhopadhyay {\etalem et al.}~\cite{MKGB} considered the related problem of computing a convex polygon transversal of minimum \emph{area} for vertical segments, giving an algorithm that runs in $O(n \log n)$ time.
Prior to the work of Arkin {\etalem et al.}~\cite{ADKMPSY} on convex transversal,
Goodrich and Snoeyink~\cite{GS} gave an $O(n\log n)$ time algorithm
that decides whether a convex transversal exists when the segments are parallel.

Pairs of points are also the input of the problems studied by Arkin {\etalem et al.}~\cite{adhkm12}, who studied the 1-center and 2-center problems in that context.
In the former problem, the goal is to find a disk of smallest radius containing at least one point from each pair. 
The latter one aims at finding two disks of smallest size such that each pair has one point in each disk. Arkin {\etalem et al.}~\cite{adhkm12} presented algorithms for these problems that run in $O(n^2 \text{polylog } n)$ and $O(n^3 \log^2 n)$ time, respectively.

Several similar problems have been considered in the context of \emph{data imprecision} by L\"{o}ffler and van Kreveld~\cite{LK,kl-alchip-08}. The input in their problems is a set of imprecise points, where each point is specified by a region in which the point may lie. The output is the location of each point within the specified region so that the area or perimeter of the convex hull is maximized/minimized. Among other cases, they consider the case in which each imprecise region is a segment~\cite{LK}. 
First, they show that for the maximization of perimeter and area, one can restrict the search to the endpoints of the regions, thus implying that their problems are equivalent to the {\sc MaxPerSPP} and the {\sc MaxAreaSPP}, respectively. 
Then, they give several polynomial time algorithms for particular cases of the input (e.g., parallel segments). They also show that the {\sc MaxAreaSPP} is NP-complete, and that {\sc MaxPerSPP} is NP-hard. In a companion paper~\cite{kl-alchip-08}, they provide a linear-time approximation scheme for the {\sc MaxAreaSPP}.

%, we cite those where regions are segments \rodrigo{Do we cite them all? smallest/largest, area/perimeter: this gives 4 different problems. But here we only talk about one of them (largest, area)?}\mati{If I understood correctly, in [14] they are considering Rappaport's definition of stabber with the 4 objective functions. They give several polynomial time algorithms, but for particular cases of the input (parallel segments, etc). In the general case, they show NP-complete for MaxArea, and NP-hard for MaxPerim.}. For maximum-area convex hulls, the problem can be solved in $O(n^3)$ time if the segments are parallel, or when they are pairwise disjoint with endpoints  in convex position. The problem is NP-hard for general segments, but admits a linear-time approximation scheme~\cite{LK,kl-alchip-08}.

%The minimum-perimeter and minimum-area convex hull problems for
%parallel segments coincide with the problems studied by Meijer and
%Rappaport~\cite{MR} and Mukhopadhyay {\etalem et al.}~\cite{MKGB},
%respectively. Notice also that the setting we consider is in
%fact a constrained version of the problems studied by L\"{o}ffler
%and van Kreveld~\cite{LK}, in which each imprecise point is
%specified by a pair of points. 

%In a more general setting,  
Daescu {\etalem et al}.~\cite{daescu2010} studied the complexity of the problem of,
given a $k$-colored point set, finding a convex polygon of minimum
perimeter containing at least one point from each color.
Note that the {\sc MinPerSPP} is the special case
in which $2n$ points are colored with $n$ colors, and each color is used twice.
They proved that their problem is NP-hard if $k$ is part of the input of the
problem, and presented a $\sqrt{2}$-approximation algorithm for the {\sc MinPerSPP} that runs in $O(n^2)$ time.

Parallel to our research, the model we consider was studied in a more general setting by Consuegra and Narasimhan~\cite{cn-ap-13} and Consuegra {\etalem et al.}~\cite{cnrt-ap-12}. They define a class of geometric problems called \emph{avatar problems}. In these problems one has a collection of objects, each of which has $k$ 
copies (or avatars). The objective is to find some structure that uses at least one copy of each object. 
Our problem fits into their model as a particular case in which $k=2$, each avatar is a point, and the structure to look for is a minimum/maximum perimeter/area convex hull of the selected points.
%Two results from~\cite{cn-ap-13} are relevant in our context.
Consuegra {\etalem et al.}~\cite{cnrt-ap-12} gave a dynamic programming algorithm that can be used to solve the {\sc MinPerSPP} in polynomial time for parallel segments.
In the companion paper~\cite{cn-ap-13}  they present a polynomial-time approximation scheme that can be applied to both the {\sc MinPerSPP} and the {\sc MinAreaSPP}. %, yielding an $O(n^8)$-time approximation algorithm.
%The latter algorithm can also be applied to minimize the area of the convex hull.

\subsection{Our results}
We show in Section~\ref{secdisjoint} that if $S$ is a set of pairwise disjoint segments,
the four variants of the {\sc SPP} can be solved in polynomial time.
%We then show how the algorithm can be adapted to solve the following maximization problem: 
%Select exactly one point on each segment in $S$ such that the perimeter (or area) of
%the convex hull of the selected points is maximized.
In particular, this method can be used to solve the open problem posed by L\"offer and van Kreveld~\cite{LK} (since their maximum area transversal problem is equivalent to our {\sc MaxAreaSPP}). In Section~\ref{sec:islands} we extend our algorithm for disjoint segments to \emph{islands of points}: $S$ is a collection of point sets with pairwise-disjoint convex hulls, proving that this problem can also be solved in polynomial time.
Finally, we show that the minimization variants of the problem for the case of general segments (that is, when crossings are allowed) is NP-hard in Section~\ref{sec:hardness}. This result complements with the NP-hardness for the maximization variants of L\"offer and van Kreveld~\cite{LK}.
%
%In Section~\ref{sec:hardness} we show that for general segments
%the \pb{MPSPP} is NP-hard.
% We also improve the result of
% Daescu {\etalem et al.}~\cite{daescu2010} showing that their problem
% is NP-hard even for fixed values of $k$ (and with simpler arguments).
We complement the NP-hardness result by showing that the four variants of the {\sc SPP} are Fixed Parameter Tractable (FPT) in the number of segments that cross other segments. A summary of the results obtained for line segments can be seen in Table~\ref{tab_res} (note that, for conciseness, our results for islands of points are not included in the table).

\begin{table}[ht] 
\centering  
\begin{tabular}{|c |c |c |} 
\hline
 &   Minimization & Maximization  \\ \hline 
 & NP-hard (Th.~\ref{thm_np_hard}) & NP-hard~\cite{LK} \\
Perimeter  & PTAS~\cite{cn-ap-13}& \\
& Polynomial for non-crossing (Th.~\ref{teo:main}) & Polynomial for non-crossing (Th.~\ref{thm:maximize})\\
& FPT (Obs.~\ref{obs_fpt}) & FPT (Obs.~\ref{obs_fpt})\\ \hline
 & NP-complete (Th.~\ref{thm_np_hard}) & NP-complete~\cite{ich-waac-11,LK}\\
Area & PTAS~\cite{cn-ap-13} & PTAS~\cite{kl-alchip-08} \\
& Polynomial for non-crossing (Th.~\ref{teo:main}) & Polynomial for non-crossing (Th.~\ref{thm:maximize})\\
& FPT (Obs.~\ref{obs_fpt}) & FPT (Obs.~\ref{obs_fpt})\\ \hline

\end{tabular} 
\vspace{3mm}
\caption{Summary of known and new results for the four variants of the {\sc Stabbing Polygon Problem (SPP)}, for a set of line segments.} % title of Table 
\label{tab_res} % is used to refer this table in the text 
\end{table} 

%throughout the paper 
Note that optimization of the perimeter requires comparing sums of radicals (specifically, the sum of Euclidean distances).
It is not known whether this problem is in NP~\cite{bloemer}, and therefore the NP-hardness result does not imply NP-completeness for the minimization version of the problem (the same fact was also observed in~\cite{LK}). For the same reason, we assume the real RAM as the underlying computational model in our algorithms.

% Due to lack of space, several proofs have been deferred to the full version~\cite{fullversion}.

\section{Solving the {\sc SPP} for pairwise disjoint segments}\label{secdisjoint}

In this section we show that if the segments in $S$ are pairwise disjoint, then the {\sc SPP} can be solved in polynomial time. For ease of exposition, we present the algorithm for the {\sc MinPerSPP}. Observe throughout the description that the approach naturally extends to the {\sc MinAreaSPP} as well.
In Section~\ref{sec:maxarea} we explain the modifications needed for the maximization variants of the problem.

Given any two points $p$ and $q$ in the plane, let $pq$ denote the segment
joining $p$ and $q$. For any 
simple polygon~$\cal P$ let $\partial \cal P$ denote the boundary of $\cal P$.
Consider the set $B$ of all possible bitangents of $S$, i.e., $B$ is the set
of all segments not contained in $S$ spanned by two endpoints of
segments in $S$. Note that the elements of $B$ might cross each
other and might also cross the segments in $S$.
A polygon $C^*$ with minimum perimeter that contains at least one
endpoint of every segment of $S$ is spanned by endpoints of
segments in $S$, and its edges are elements of $B$.

Arkin {\etalem et al.}~\cite{ADKMPSY} describe a dynamic programming approach to decide whether a set of pairwise disjoint segments admits a convex transversal (the vertices of the transversing polygon are restricted to a given set of candidate points).
They use constant-size polygonal chains that separate subproblems and are not crossed by segments; therefore the subproblems are independent.
We adapt their approach to produce an algorithm for the {\sc MinPerSPP}.
While in their problem setting the boundary of the
polygon has to intersect all segments, the {\sc SPP} requires at least
one endpoint of each segment to be contained in the polygon.
The key difference (apart from the fact that no candidate points are needed) is that in our problem the segments actually \emph{can} cross the polygonal chains that separate subproblems.
However, we show below that such segments can be handled in a way that leads to polynomial running time.
% We apply dynamic programming to find $C^*$, extending an approach by Arkin {\etalem et al.}~\cite{ADKMPSY}.
% They use constant-size polygonal chains that separate subproblems and are not crossed by segments.
% The latter property is in contrast to our approach where segments actually cross this separating chain but can be handled in a way that allows overall polynomial running time.
Afterwards, we discuss how to adapt this approach for the maximization variant.

\subsection{Triangulating a combination of segments and a polygon}
The following way of triangulating a combination of segments and a
polygon is crucial for the algorithm, and motivates the structure
of the subproblems used in our dynamic programming algorithm.

Let $\mathcal{Q}$ be a simple polygon and let $S_c$ be a set of pairwise disjoint segments each of which crosses $\partial \mathcal{Q}$ exactly once.
Throughout this section we distinguish between a segment \emph{intersecting} (having a point in common) and \emph{crossing} (having an \emph{interior} point in common with) another segment or set.
Let $X$ be the interior of $\mathcal{Q}$ and let $X'$ denote the set we get
after removing the $1$-dimensional regions of $S_c$ from $X$, i.e.,
$X'=X\setminus \bigcup_{s\in S_c}s$ where each segment~$s$ is considered to be an infinite set of points.
Then $X'$ is an open region whose closure is
$\mathcal{Q}$.
Note that the vertices of $X'$ are the union of:
(i)~the vertices of $\mathcal{Q}$, (ii)~the endpoints of edges in $S_c$ that are in the interior of $\mathcal{Q}$, and
(iii)~the points where elements of~$S_c$ cross $\partial \mathcal{Q}$.
Further, note that $X'$ might not be connected if there is a segment
of $S_c$ that has one endpoint on $\partial \mathcal{Q}$ and the other one outside  $\mathcal{Q}$ (e.g., the longest segment in \figurename~\ref{fig_nc_triangulation}, left).

We now triangulate $X'$ (i.e., partition it into triangles that are spanned only by vertices of $X'$, see \figurename~\ref{fig_nc_triangulation}).
The triangulation $T$ of $X'$
behaves like the triangulation of a collection of simple polygons (imagine the $1$-dimensional
parts not in $X'$ where the segments of $S_c$ enter $\mathcal{Q}$, i.e., $X\setminus X'$, to be slightly ``split'', as in \figurename~\ref{fig_nc_triangulation}, center).
Note that the vertices of $T$ are exactly the vertices of~$X'$.
Each edge in~$T$ that is not part of $\partial \mathcal{Q}$ or part of a segment in $S_c$ partitions $X'$ into two sets (note that each set need not be connected).
We call such edges \emph{chords} (gray edges in \figurename~\ref{fig_nc_triangulation}, right).
Hence, a chord is an edge where each endpoint is either an endpoint of a segment of~$S$ or a crossing between a segment of~$S$ and a bitangent of~$B$.
Chords are the equivalent of \emph{diagonals} of simple polygons (interior edges that subdivide the polygon into two smaller polygons).
Further, $X'$ might also be separated by an edge that is part of a segment in $S_c$ (like the longest edge in \figurename~\ref{fig_nc_triangulation}).
We call such a segment a \emph{separating segment}.
Keep in mind that there are chords that have one or both of their endpoints not on the endpoint of a segment or at a vertex of $\mathcal{Q}$, but at the crossing of a segment with $\partial \mathcal{Q}$.
In any case, a chord or a separating segment uniquely defines a polygonal path from one point on an edge of $\mathcal{Q}$ to another point on an edge of~$\mathcal{Q}$. Following~\cite{ADKMPSY}, we will use these polygonal paths of at most three edges, called ``bridges'' (whose formal description will be given later), to define our subproblems to obtain a solution when taking $C^*$ as $\mathcal{Q}$. Further note that at most two of the edges are a portion of a segment of $S$. One may think of our approach as being similar to the classic dynamic programming algorithm for minimum weight triangulations of simple polygons~(see, e.g.,~\cite{klincsek}), but with a major difference:
we do not know the boundary of the triangulated region beforehand.

\begin{figure}[bt]
\centering
\includegraphics{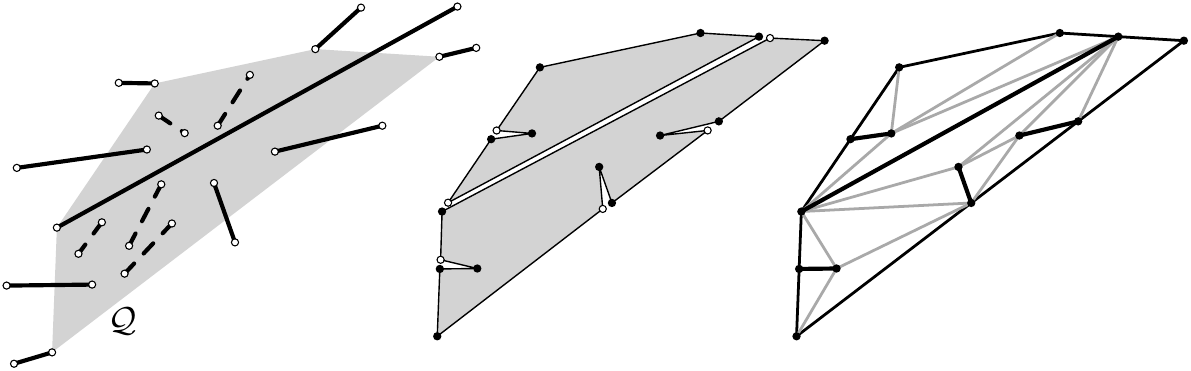}
\caption{Left: an optimal polygon $\mathcal{Q}$, only the solid edges are in $S_c$.
Center: schematic view of $X'$ as a collection of simple polygons.
Right: a triangulation of $X'$, gray edges are chords. 
The segments fully contained in the polygon (shown dashed) are ignored by the triangulation.}
\label{fig_nc_triangulation}
\end{figure}

\subsection{Subproblems}
Every subproblem is defined by an ordered pair $(a,b)$ of directed bitangents
of $B$ and a \emph{bridge} $\beta$, a polygonal chain of at most three edges that connects $a$ and $b$. 
When evaluating a subproblem $(a,b,\beta)$, we assume that
$a$ and~$b$ are edges of $C^*$ (with $a$ being directed counterclockwise and $b$ being directed clockwise around $\partial C^*$) and that $C^*$ equals $\mathcal{Q}$ in the discussion above (for some choice of~$S_c$ to be defined later).
Therefore the bridge $\beta$ is part of a triangulation of~$X'$ and separates $X'$;
$\beta$ is either a part of a separating segment or consists of a chord (called the \emph{chord of $\beta$})
and at most two parts of segments of $S_\mathrm{c}$.
See \figurename~\ref{fig_bridge_types} for examples of bridges.

Let us recap the possible structures of bridges.
Traversing a bridge~$\beta$ from $a$ to $b$, $\beta$ starts from either (i) an endpoint of $a$,
% (that therefore is an endpoint of a segment in~$S$) 
or (ii) the intersection point of some segment $s \in S$ and bitangent~$a$.

In the first case, when $\beta$ starts from an endpoint of $a$, $\beta$ consists of a separating segment ending at its intersection point with bitangent~$b$, or~$\beta$ contains a chord that connects to an endpoint of $b$ or to a piece of a segment that crosses~$b$.

In the second case, when $\beta$ starts from intersection point $s \cap a$, the bridge either continues with a chord, which starts at $s \cap a$, or it continues along $s$.
In the latter case, $\beta$ continues along $s$ towards $b$ until reaching its endpoint.
The bridge can end there, if that endpoint is also an endpoint of bitangent~$b$ (in which case $s$ is a separating segment) or it continues through a chord that connects to $b$ or to a piece of a segment that crosses~$b$.

The analogous structure occurs when traversing~$\beta$ from $b$ to $a$.
Keep in mind that a bridge might have a chord that is not a bitangent of $B$ (like the second from the left in  \figurename~\ref{fig_bridge_types}).
Further, note that a bridge can only be crossed by a segment through
the chord, since the segments are pairwise disjoint by definition.

\begin{figure}[bt]
\centering
\includegraphics{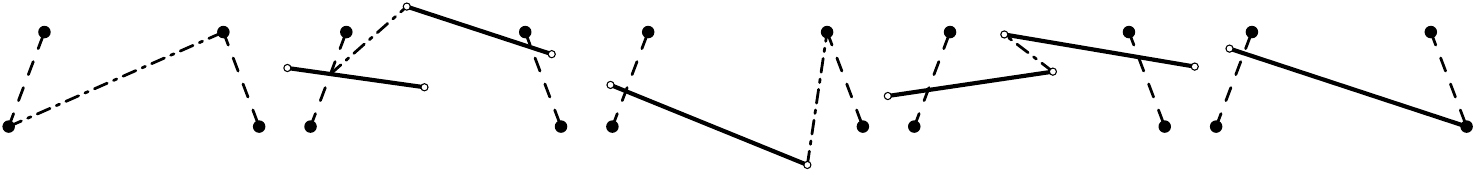}
\caption{Examples of bridges. The two bitangents defining the subproblem are shown dashed, chords are dash-dotted, and segments from $S_c$ are shown solid.}
\label{fig_bridge_types}
\end{figure}

Let the two directed bitangents of a subproblem be $a=a_1 a_2$ and $b=b_1 b_2$.
Given a directed bitangent $a=a_1 a_2$ we write $\overline{a}$ for
the directed bitangent $a_2 a_1$. Without loss of generality, let $a_1$ and $b_1$ be
on the $x$-axis and $a_2$ and $b_2$ be above it.
Also, let $b$ be to the left of the directed line through $a_1$ and $a_2$.
See \figurename~\ref{fig_subproblem_examples} for an illustration.

%We define the outcome of a solved subproblem as follows.

% Let $C^*_{a,b}$ be a polygon of minimum perimeter that contains the bitangents $a$ and $b$ as two of its edges and contains at least one endpoint of each segment in $S$.
% Note that $C^*_{a,b} = C^*_{\overline{b},\overline{a}}$.
% Let $C_{a,b}$ be the polygonal chain on $\partial C^*_{a,b}$ starting at $a_1$, counterclockwise traversing $\partial C^*_{a,b}$ and ending at~$b_1$.
% Note that $C_{a,b} \neq C_{\overline{b},\overline{a}}$.

\subsection{Solution of a subproblem}
We define the solution of a subproblem as follows.
Let $C^*_{a,b,\beta}$ be a polygon of minimum perimeter that: 
(i) contains $a$ and $b$
as two of its boundary edges, 
(ii) contains at least one endpoint of each segment in $S$, and
(iii) contains both endpoints of every segment of $S$ that crosses the chord of~$\beta$.
\ShoLong{%
The importance of the third condition will become clear later.
}
{%
The third condition is particularly important, as will become clear later.
}

Let $C_{a,b,\beta}$ be the polygonal chain on $\partial C^*_{a,b,\beta}$ starting at
$a_1$, counterclockwise traversing $\partial C^*_{a,b,\beta}$ and ending at $b_1$.
Note that $C_{a,b,\beta}$ is an open polygonal chain, as opposed to $C^*_{a,b,\beta}$, which is a simple polygon.

The solution of a subproblem $(a,b,\beta)$ is $C_{a,b,\beta}$, and its cost is the
length of that chain. 
The base case occurs when $a_2=b_2$, and has cost equal to the sum of the lengths of $a$ and~$b$.
% \mati{Alex said that  a1=b1 or a1=b2 cannot happen. Can you briefly justify this? It is basically that the associated solution has cost infinity because it does not form a correct quadrilateral?}
% \alex{yes}
Note throughout the construction that this is the only way $a$ and $b$ can intersect.
In general, $a$ and $b$ form a quadrilateral $ a_2 a_1 b_1 b_2$.
If the quadrilateral is not convex, we discard the subproblem (i.e., we assign it a cost of $+\infty$).
Therefore, from now on we discuss the more interesting case in which the quadrilateral is convex.

\begin{figure}[bt]
\centering
\includegraphics[width=\textwidth]{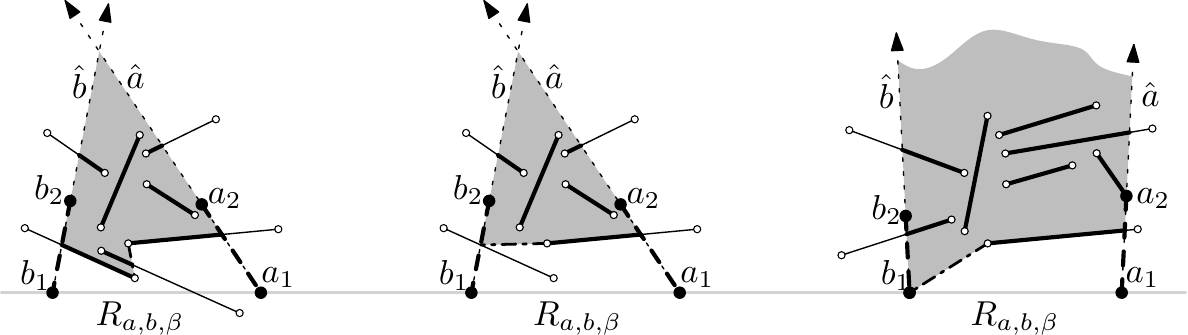}
\caption{Examples of subproblems.}
\label{fig_subproblem_examples}
\end{figure}

\subsection{Getting the overall solution}
%From now on we assume that  $a$ and $b$ define a convex quadrilateral.
%The outline of the algorithm is as follows.
In order to find a solution of the initial problem we need to find $a,b, \beta$ so that the solution to the 
subproblem $(a,b,\beta)$ gives a solution of the given instance of the {\sc MinPerSPP}.
We do that by guessing 
a pair of bitangents $x,y \in B$, with $x=x_1x_2, y=y_1y_2$, such that
$y_2y_1x_1x_2$ are assumed to be four consecutive vertices of $C^*$. Hence, after $O(|S|^4)$ guesses
we have found $x$ and $y$ such that $\partial C^* =C_{x,y,\beta_0}\cup y_1 x_1$ with
$\beta_0 = x_1 y_1$.
%
%\mati{rather than an outline this is an initialization. Should we write so? And start the outline section immediately afterwards?}
%\alex{Well, an outline will also contain the initialization. If I remember correctly, I was even asked to move this part here.}
%
%
%The subproblem $(x,y,\beta)$ will be solved in one of several different ways.
%In general, solving them will involve finding a bitangent $c$, generating two new bridges $\beta_1$ and $\beta_2$ and combining the solutions for the subproblems $(x,\overline{c},\beta_1)$ and $(c,y,\beta_2)$.
%
%
Suppose we are given the solution $\mathcal{Q}=C^*$.
Let $X'$ be defined as above, and let $S_c$ be the set of segments in $S$ that cross $C_{x,y,\beta_0}$ (which does not include the ones that cross $\beta_0$). 
Let $\Delta_0$ be the triangle of a triangulation $T$ of $X'$
that has $\beta_0 = y_1 x_1$ as one side. The subproblem $(x,y,\beta_0)$ will be solved
by guessing the third endpoint of $\Delta_0$ and the edge $c$ of $C_{x,y,\beta_0}$ that is incident to $\Delta_0$ or that is
crossed by a segment whose endpoint is incident to~$\Delta_0$. In the most general case,
this will result in two new subproblems $(x,\overline{c},\beta_1)$ and $(c,y,\beta_2)$, where each of
$\beta_1$ and $\beta_2$ contains one side of $\Delta_0$ that is not part of $\beta_0$
(we will consider the other cases in detail below, as well as the exact rules for guessing the third endpoint). See \figurename~\ref{fig_nc_start}.

% Note that even though $c$ is part of the $C_{x,y,\beta_0}$, the guess for $\Delta_0$ shown in the figure is still wrong.
% However, a different guess results in the optimal solution.

\begin{figure}[ht]
\centering
\includegraphics{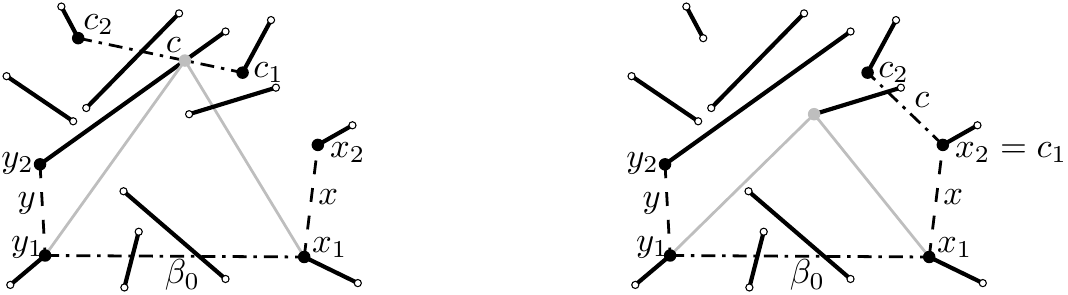}
\caption{An example for choosing the initial pair defining a subproblem and two choices for~$\Delta_0$ (gray). We will observe that the choice of $\Delta_0$ in the example to the left does not result in the optimal solution, even though $c$ is an edge of the optimal solution.}
\label{fig_nc_start}
\end{figure}
% \mati{I find the term  "initial pair" a bit misleading (I thought it referred to the first query we ask to the DP problem). How about "Situation for a typical cabbbeta subproblem"? or something similar?. }

\subsection{Subproblem regions}
Let $\hat{a}$ be the ray through $a_2$ starting at $a_1$.
Let $\hat b$ be defined analogously.
For every subproblem $(a,b,\beta)$, only a part of the elements of $S$ is relevant.
Consider the (possibly unbounded) maximal region to the left of the supporting line of $a$ and to the right of the supporting line of $b$ (recall that $a$ and $b$ are directed).
The bridge $\beta$ disconnects that region into two parts.
The \emph{subproblem region} $R_{a,b,\beta}$ is the part ``above'' $\beta$
(i.e., the part adjacent to $\hat{a} \setminus a$ and $\hat{b} \setminus b$;
the bridge might not be $x$-monotone).

The subproblem region is marked gray in 
\figurename~\ref{fig_subproblem_examples}.
Only the segments that have at least one endpoint in $R_{a,b,\beta}$ are
relevant for finding $C_{a,b,\beta}$.
We distinguish between three different types of such segments:
(1) Segments that are entirely inside $R_{a,b,\beta}$ are \emph{complete}.
(2) Segments that share more than one point with $R_{a,b,\beta}$ but are not complete are \emph{cut}.
(3) A segment with infinitely many points on the bridge is neither cut nor complete.
We say that a point is \emph{inside} $C_{a,b,\beta}$ when it is contained in the closure of
the region bounded by $C_{a,b,\beta}$ and $\beta$.

If there is a segment that is entirely to the right of
$a$ or to the left of $b$, then the choice of $a$ and~$b$ cannot give a solution and such a subproblem is assigned~$+\infty$ as cost.
We also do this if a segment intersected by $\hat{a}$ or $\hat{b}$ does not have an endpoint inside the
subproblem region.
% \alex{I slightly changed this last sentence, check please.}

Note that if a segment in a valid subproblem intersects $\hat a$ or $\hat b$, then we know which of its
endpoints must be inside $C_{a,b,\beta}$, while we do not know that for the cut segments that intersect the chord of the bridge. However, we will choose our subproblems in a way such that all endpoints of cut
segments in the subproblem region will be inside $C_{a,b,\beta}$; the reason
for that will become clear in the proof of Lemma~\ref{lem_case_4}, but the reader should keep this in mind as an
essential part of the method. For complete segments, we need to decide which
endpoint to select.

\begin{lemma}\label{lem_c_triang}
Given a subproblem instance $(a,b,\beta)$, let $t$ be the chord of $\beta$, or its only edge if $\beta$ is a single edge (which may be a chord itself, or part of a separating segment).
Let $X$ be the region bounded by $C_{a,b,\beta} \cup \beta$, and let $X' = X \setminus \bigcup_{s\in S_\mathrm{c}}s$, for $S_c$ the set of segments of $S$ that are crossed by the chain $C_{a,b,\beta}$.
Then either $t$ is an edge of $C_{a,b,\beta}$, or there exists a triangle $\Delta$ such that:
\begin{enumerate}
%\begin{itemize}
 \item The interior of $\Delta$ is completely contained in $X'$.
 \item The edge $t$ is an edge of $\Delta$.
 \item The apex of $\Delta$ (i.e., the vertex not on $t$) is either (i) an endpoint of a segment in $S_\mathrm{c}$ inside $X$, (ii) an endpoint of a segment in $S$ that is a vertex of $C_{a,b,\beta}$, or (iii) an intersection point between a segment in $S_c$ and $C_{a,b,\beta}$.
%\end{itemize}
\end{enumerate}
\end{lemma}

%\ShoLong{}
{\begin{proof}
Arbitrarily triangulate $X'$.
If $t$ is not on the boundary, then the triangle~$\Delta$ incident to $t$ inside the subproblem region fulfills the properties.
See \figurename~\ref{fig_nc_lem_c_triang}.
\end{proof}
}

\begin{figure}[ht]
\centering
\includegraphics{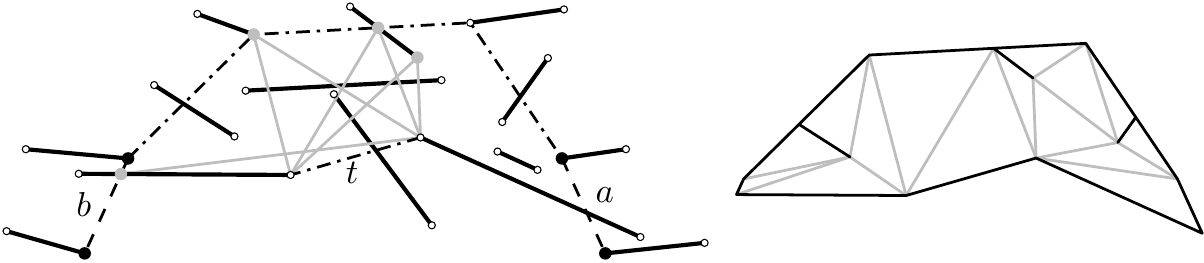}
\caption{Illustration of Lemma~\ref{lem_c_triang}.
Left: four possibilities for $\Delta$ shown in gray.
$C_{a,b,\beta}$ is dash-dotted, with the defining bitangents dashed.
Right: a triangulation of $X'$.}
\label{fig_nc_lem_c_triang}
\end{figure}

\begin{lemma}\label{lem_both_inside}
Let $\Delta$ be the triangle of Lemma~\ref{lem_c_triang}.
Any segment of $S$ that has a non-empty intersection with the interior of
$\Delta$ either has both its endpoints inside $C_{a,b,\beta}$ or crosses $t$;
in the latter case the endpoint that is inside $R_{a,b,\beta}$ is also
inside $C_{a,b,\beta}$.
\end{lemma}

%\ShoLong{}{
\begin{proof}
This follows from the properties of $\Delta$ in Lemma~\ref{lem_c_triang}:
A segment intersecting the interior of $\Delta$ is not part of $S_\mathrm{c}$
but has a non-empty intersection with $X$. Therefore, either both of its endpoints
are inside $C_{a,b,\beta}$, or it enters $X$ via $t$ and therefore has its relevant
endpoint inside $C_{a,b,\beta}$ by definition. See \figurename~\ref{fig_nc_lem_c_triang}.
\end{proof}
%}

% Recall that for the choice of the first bridge $y_1 x_1$ of the algorithm, it is clear that the corresponding endpoints of the segment that cross the bridge have to be inside the solution.
% When $\Delta$ is chosen as the triangle adjacent to the bridge, this requirement carries over to the subproblems.
%

\subsection{Getting smaller subproblems}\label{sec_smaller_subproblems}
%
%\rodrigo{I think the next paragraph should be somewhere else, here it does not fit very well. Maybe in the next subsection, since it defines $N$.}
%
Let $A$ be the set of points that are either endpoints of $S$ or crossing points of a segment
and a bitangent (recall that no segment of $S$ is an element of $B$).
%
%
% \rodrigo{This is a technicality, but it pops up in many places, and sooner or later we have to fix it: $\cup_{i \in S}i$ is a set of segments (not of points in $R^2$, so the intersection with $\cup_{j \in B}$ gives the segments that are both in $S$ and in $B$ (in our context, none). I know that we mean the points contained in the segments in $S$ (resp. $B$), but that is not what is written. I think that keeping an informal description (like the one after the ``i.e.'' before) is better than writing up these things  formally.}
% \alex{Yes, I changed to the verbal definition.}
%
Hence, $A$ contains all the points that are possible apices for a triangle
$\Delta$ of Lemma~\ref{lem_c_triang}. 
% \ShoLong{Note that $|A| \in O(|S|^3)$ since $|B| = 4\binom{|S|}{2}$.}
{Note that one may construct subproblems where every possible apex
of $\Delta$ is an endpoint of a segment in $S_\mathrm{c}$, as well as
subproblems where every possible apex is on a point where a segment
crosses $C_{a,b,\beta}$. Further, note that $|A| \in O(|S|^3)$ since $|B| = 4\binom{|S|}{2}$.}

Consider again a subproblem $(a,b,\beta)$. 
If $a_2 = b_2$, then we have reached the end of the recursion, and there are no smaller subproblems to consider.
Otherwise, as in Lemma~\ref{lem_c_triang},
let $t$ be the chord of $\beta$ if a chord exists, or otherwise let $t$  be
the only edge of $\beta$.
%
% \rodrigo{Can't it have only 2 edges?}
% \alex{Yes, that was a mistake.}
%
Let $a_\beta$ be the intersection point of $a$ with the bridge $\beta$;
$b_\beta$ is defined analogously. 
For each subproblem $(a,b,\beta)$, one of the following cases applies, allowing to obtain one or two smaller subproblems.
%that is not a base case (i.e., $a_2 \neq b_2$) 
During the execution of the algorithm we will consider both cases.
% (since we will not know in which of the two we are).

{\bf Case 1: $\bm{t}$ is an edge of the solution, i.e., an edge of $\bm{C_{a,b,\beta}}$.}
This happens when $t$ is a chord that does not intersect the interior of the quadrilateral defined by $a$ and $b$.
%
%
% %\rodrigo{Is this the same as: all remaining cases except Case 3?}
% %\alex{I don't see what you mean. What has it to do with Case~3? The chord of the bridge can be a bitangent or also (like for Type~3) just a segment (not in $S$) between some segment endpoint and a crossing between a segment and a bitangent.}
% %\rodrigo{Isn't Case 3 the only case where the chord may not be a bitangent? If that is so, then your sentence `consider the possibility that $t$ itself is a bitangent' would be the same as `consider cases different from Case 3'. Maybe saying that could help the reader understand what's going on...}
% %\alex{Hm, I guess one could also remove Case 2, as well as Case 1 if one then says that a subproblem is finished as soon as the endpoints of the two defining bitangents meet.
% %Well, actually, one could also say that if one has Case 3, one takes the part of the segment that forms the bridge as a chord.
% %The cases would speed up the process since there is less choice, but on the other hand, it might not change anything with the asymptotic running time.
% %Let's discuss this today.}
%
%
This case is only valid if no segment crosses~$t$, as we require all the
endpoints in $R_{a,b,\beta}$ of segments crossing $t$ to be inside $C_{a,b,\beta}$.
In that case we get at most two new subproblems $(a,\overline{t},\beta_1)$ and
$(t,b,\beta_2)$, where $\beta_1$ is the edge $a_\beta t_1$ and $\beta_2$ is
the edge $t_2 b_\beta$. However, note that one of $(a,\overline{t})$ or $(t,b)$
(or both) might intersect at $a_2$ or $b_2$, respectively, and therefore form
a base case.

%
% If there exists a segment that crosses $t$ and $t$ is an edge of $C^*$, then the solution is composed of different subproblems.
% \rodrigo{I could not understand the last sentence}.
% \alex{I removed it, the thing it should mean actually clearly follows from the whole reasoning up to now}
%

{\bf Case 2: $\bm{t}$ is not an edge of the solution.}
Then there is a triangle adjacent to $t$ as in Lemma~\ref{lem_c_triang}. 
We will guess the apex of the triangle.
For every point $d$
in $A \cap R_{a,b,\beta}$ consider the triangle $\Delta_d$ that $d$ forms with~$t$.
We only consider $d$ if $\Delta_d$ is completely inside $R_{a,b,\beta}$, and
 the interior of $\Delta_d$ does not intersect any segment that intersects $a$ or~$b$.
It follows from Lemma~\ref{lem_c_triang}
that one of the triangles tested leads to a subdivision of the optimal solution.
We get the following two subcases,
see \figurename~\ref{fig_nc_case4}.

{\bf Case 2.1: $\bm{d}$ is a point where a bitangent and a segment cross.}
Let $c$ be the bitangent that contains~$d$. If $c$ equals $a$ or $b$,
then we get one new subproblem $(a,b,\beta')$, with~$\beta'$ containing
a side of $\Delta_d$ that is different from~$t$ as a chord (\figurename~\ref{fig_nc_case4}a). Otherwise, we get two new subproblems,
$(a,\overline{c},\beta_1)$ and $(c,b,\beta_2)$, where $\beta_1$ and $\beta_2$
both contain a side of~$\Delta_d$ (\figurename~\ref{fig_nc_case4}b).

{\bf Case 2.2: $\bm{d}$ is an endpoint of a segment.}
Let $s$ be the segment that has $d$ as its endpoint. Choose a point $x$
where $s$ intersects some bitangent $c$. Then, for every possible choice
of $x$ (which implies the choice of $c$), we get two new subproblems
$(a,\overline{c},\beta_1)$ and $(c,b,\beta_2)$, as in the previous case;
note that for both new bridges, $x=d$ is possible. The degenerate case
where $c$ equals $a$ or $b$ can be handled as in the previous case.
See \figurename~\ref{fig_nc_case4}c-d.

\begin{figure}[tb]
\centering
\includegraphics{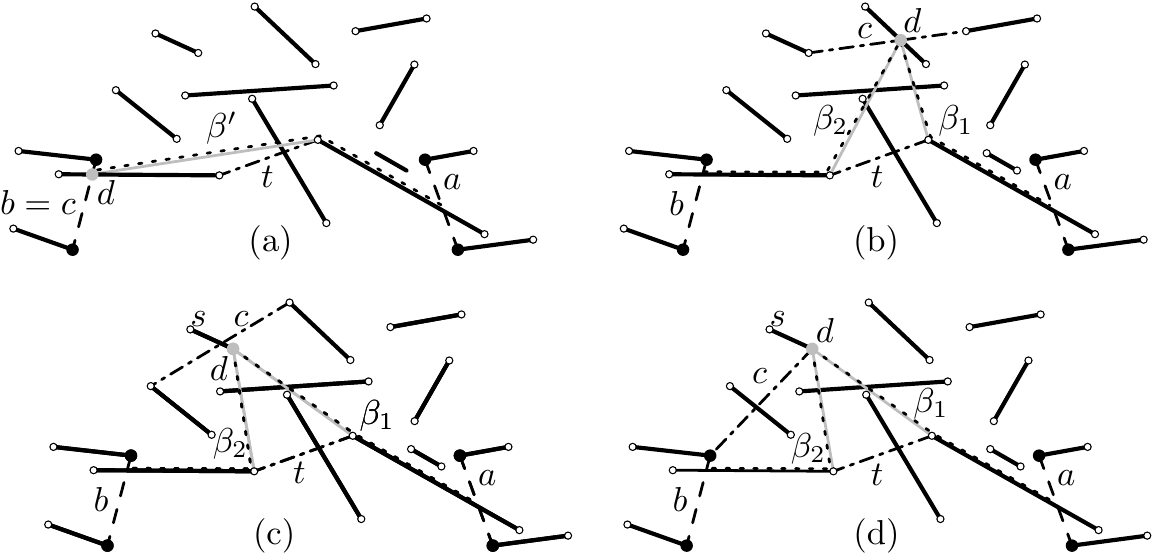}
\caption{Case~2.
The new bridges are dotted. (a)-(b) Case 2.1. (c)-(d) Case 2.2.}
\label{fig_nc_case4}
\end{figure}

In all cases, after considering each case and the associated subproblems, we compute the information about the perimeter of the current solution accordingly.

\begin{lemma}\label{lem_case_4}
Given any valid subproblem $(a,b,\beta)$, there is a pair of subproblems among the ones 
above such that the union of their solutions is equal to~$C_{a,b,\beta}$.
\end{lemma}

\begin{proof}
Consider the edge~$t$ of Lemma~\ref{lem_c_triang}.
If $t$ is a chord and part of $C_{a,b,\beta}$, then it will be considered in Case~1.
Otherwise, consider the triangle $\Delta$ inside
$C_{a,b,\beta}$. All segments that are intersected by the interior of $\Delta$ are
either completely contained in $C_{a,b,\beta}$ or enter through $t$ (if it is a chord)
and therefore have their relevant endpoint inside $C_{a,b,\beta}$
(cf.\ Lemma~\ref{lem_both_inside}).
Hence, when the choice of $\Delta_d$ coincides with~$\Delta$, the two subproblems
can be combined into $C_{a,b,\beta}$; the only segments that are part of both subproblems
intersect the interior of~$\Delta$, and we know that both
endpoints will have to be inside the chain that results from the combination of the
solutions of the subproblems. For each possible value of $\Delta_d$ we obtain a stabber, its cost cannot be lower than the one of the optimal solution. Moreover, since we check all possibilities of~$\Delta_d$,
the subproblem combination of minimum cost is guaranteed to be~$C_{a,b,\beta}$.
\end{proof}

This last lemma now implies that we actually find the optimal solution.
% \ShoLong{}
{Note that it is easy to construct a pair of bitangents and a bridge $(a,b,\beta)$
that is part of the optimal solution but for which $C_{a,b,\beta}$ is not part of $C^*$.
However, as mentioned in the outline of the algorithm, we choose the initial problem
$(x,y,\beta_0)$ in a way that $\partial C^* = C_{x,y,\beta_0} \cup \beta_0$.
All segments crossing $\beta_0 = x_1 y_1$ need to have their endpoint above $\beta_0$
inside the solution, and the algorithm actually produces a triangulation of $X'$ when taking
$C^*$ as $\mathcal{Q}$ and $S_c$ being the segments that cross $\partial C^*$ but do not cross $\beta_0$.}

Recall that we initialize the algorithm using a brute-force approach: that is, we consider all the $O(|S|^4)$ possible choices for two defining bitangents and a bridge $a_1b_1$. 
Every subproblem contains less edges of the complete graph on all endpoints of $S$,
and for every subproblem we need polynomial time.
The number of subproblems can be bounded by the choices for $c$ and $d$.
Therefore, dynamic programming can be applied to obtain a polynomial-time algorithm.\footnote{A straightforward analysis of the running time results in $O(|S|^9)$, which probably can be improved. In any case, we consider that our main contribution is that the problem can be solved in polynomial time.}%, more than the running time itself.}

\begin{theorem}\label{teo:main}
Given  a set of pairwise disjoint segments, both the {\sc MinPerSPP} and the {\sc MinAreaSPP} can be solved in polynomial time.
% a Minimum Perimeter Stabbing Polygon (i.e., a minimum perimeter polygon containing at least one endpoint of each segment in $S$) can be computed 
\end{theorem}

%Since our algorithms are combinatorial and only the cost function depends on the geometry of the problem instance, the methods of this Section extend to the other three variants of the problem. 

%
%\rodrigo{Shall we say that we use the real RAM model somewhere in the beginning and remove it from the theorems?}

% Interestingly, the part that relies on the fact that the segments are
% pairwise disjoint is not that obvious. Observe that if two segments would cross,
% we would never consider the case where both are intersected by $C$,
% since if one is in the subproblem defined by a bitangent~$c$ through
% the other, the endpoints are in different subproblems and would both be added.
% This is not a problem with the segments crossing the bridge, as for these the
% subproblems are chosen appropriately.
% 
% \rodrigo{I'm afraid I could not follow the previous paragraph. Would it be possible to put a very small example illustrating how things change when there are crossings?}

%\paragraph*{Remark}

% The running time of the algorithm is highly influenced by the fact that we have
% little information on the optimal solution. For some particular subproblems
% (i.e., when the endpoints of $S$ are in convex position, or when there is a
% line that stabs all segments of $S$), we can speed up our dynamic programming
% algorithm so that it runs in $O(n^2)$ time. Due to space constraints, details
% for these cases are omitted (and will be given in the extended version
% of this paper).

%\paragraph*{Maximization for non-crossing segments}

\subsection{Maximization for pairwise disjoint segments.}\label{sec:maxarea}
Our previous algorithm relies on the fact that the result has minimum perimeter (or area):
this automatically prevents two endpoints of the same segment from being vertices of
the resulting polygon. However, making the algorithm slightly more sophisticated,
we can solve in polynomial time maximization versions of these problems, stated open by L\"offler and van Kreveld~\cite{LK}: 
{\em select exactly one point on each segment in $S$ such that the perimeter (or area) of
the convex hull of the selected points is maximized}.
This result is based on the fact that for the maximum area or perimeter transversal, one needs to consider 
only the endpoints of the segments~\cite[Lemmata~1 and~8]{LK}:

\begin{lemma}[L\"offler, van Kreveld]
\label{lem:LK}
The problem of, given a set of line segments, choosing one point on each line segment
such that the perimeter (or area) of the convex hull of the resulting point set is as large as possible, has a solution in which all chosen points are endpoints of the line segments.
\end{lemma}

\begin{theorem}\label{thm:maximize}
%Let $S$ be a set of pairwise non-crossing line segments.
Given  a set of pairwise disjoint segments, the {\sc MaxPerSPP} and the {\sc MaxAreaSPP} can be solved in polynomial time.
\end{theorem}

\begin{proof}
Due to Lemma~\ref{lem:LK}, we know that we only need to consider the endpoints of the segments.
We modify the algorithm used for the minimization version of the problem.
Note that the structure of the solution is very similar.
Again, let $C^*$ be the optimal solution.
One main difference is that a segment that has an endpoint as a vertex of $C^*$ might have the other endpoint in the interior of $C^*$, i.e., might be completely contained in it.
We define subproblems and bridges in the same way.
The crucial property in the previous algorithm was that a segment that entered a subproblem region through the chord of the bridge had to have its endpoint that was inside the subproblem region to be inside the solution of the subproblem as well.
This was due to Lemma~\ref{lem_both_inside} and the choice of the initial bridge $\beta_0$;
a segment that enters the subproblem region through the chord of the bridge can be of one of two types:
it either crosses~$\beta_0$, thus one of its endpoints is considered to be oustide of~$C^*$,
or it does not cross~$\beta_0$, and thus both of its endpoints are considered to be in the interior of~$C^*$ (recall Lemma~\ref{lem_both_inside}).
For Theorem~\ref{teo:main}, it was not necessary to distinguish between these two types of segments (in the minimization version, a construction using both endpoints of a segment would be considered valid, but the minimum solution would never contain two such points).
However, now we need to take this into account.

Instead of only guessing three consecutive bitangents that form the initial problem $(a,b,\beta_0)$, we may choose two ``opposite'' bitangents in the following way:
For every segment~$s$, guess two bitangents $a$ and $b$ such that $a$ crosses $s$ and $b$ has a common endpoint with $s$.
This defines two subproblems $(a,b,\beta)$ and $(\overline{b}, \overline{a}, \beta)$, where $\beta$ is the part of $s$ connecting $a$ with $b$, which can be combined to an overall solution;
see \figurename~\ref{fig_nc_maximize_initial}.
We call this the \emph{first phase} of the algorithm.
Afterwards (in the \emph{second phase}), we guess three consecutive bitangents to form $(x,y,\beta_0)$ as before.
All endpoints of segments crossing the chord of a bridge then have to be in the interior of the solution.
We explicitly do not allow the solution to a subproblem to have an endpoint of a segment that crosses the current bridge as a vertex.
Hence, there might be subproblems for which no solution is possible in that way.
All solutions that would contain an endpoint of a segment entering through~$\beta_0$ as a vertex are already found when during the first phase for the following reason.
Suppose there would be a solution with a vertex~$w$ being the endpoint of a segment~$e$ crossing~$\beta_0$.
Then during the first phase we already guessed a bitangent~$a$ that equals~$\beta_0$ (as $\beta_0$ is also a bitangent), and a bitangent~$b$ incident to $w$, with the segment~$e$ being~$s$ and vertex~$w$ being~$v$ (see again \figurename~\ref{fig_nc_maximize_initial}).
Hence, such a solution was already found in the first phase, and the only solutions we still need to consider are the ones where the endpoint of any edge~$e$ crossing $\beta_0$ is not a vertex of the convex hull.

Recall the proof of Lemma~\ref{lem_c_triang}.
If we replace $S_c$ by the set of the segments that \emph{intersect} $C_{a,b,\beta}$, the analogous result follows. Following the proof of Lemma~\ref{lem_both_inside}, we observe that the segments not in $S_c$ have both endpoints in the interior of the solution.
Therefore, the choice of the bitangent $c$ that gives new subproblems for a subproblem $(a,b,\beta)$ can be altered in the following way.
If $c$ shares an endpoint with a segment that has its other endpoint on $a$ or~$b$, then $c$ is not valid.
Further, $c$ must not share an endpoint with a segment that crosses~$\beta$ (however, the requirement that all endpoints in $R_{a,b,\beta}$ of segments that cross $\beta$ have to be inside the subproblem solution persists).
% \rodrigo{Is ``$c$ contains a point $x$'' the same as ``an endpoint of $c$ is point $x$? If find talking about ``containing'' rather confusing here}

Our modification therefore only concerns the selection of $c$ in Case~2.
In both subcases, the choice for the bitangent~$c$ is restricted to the ones that do not share an endpoint with a segment crossing~$\beta$, and that do not share an endpoint with a segment sharing the other endpoint with $a$ or $b$.
In Case~2.2, we have more potential candidates for~$c$: the point~$x$ can also be the endpoint of the segment that is not~$d$ (recall that the solution might completely contain a segment that contributes a vertex to it), in which case $c$ is a bitangent that has~$x$ as an endpoint.
With this variation, we never select both endpoints of a segment but still find (a triangulation of) the optimal solution.
\end{proof}

% 
% 
% In the following we describe an alternative method to initialize the algorithm. This approach will be useful in the upcoming section. %There is a variant of the algorithm described in the proof of Theorem~\ref{thm:maximize} that will be useful in connection with the next section.
% Instead of guessing three consecutive bitangents that form the initial problem $(a,b,\beta_0)$, we may choose two ``opposite'' bitangents in the following way:
% For every segment~$s$, guess two bitangents $a$ and $b$ such that $a$ crosses $s$ and $b$ has a common endpoint with $s$.
% This defines two subproblems $(a,b,\beta)$ and $(\overline{b}, \overline{a}, \beta)$, where $\beta$ is the part of $s$ connecting $a$ with $b$, which can be combined to an overall solution;
% see \figurename~\ref{fig_nc_maximize_initial}.
% Afterwards, guess three consecutive bitangents to form $(a,b,\beta_0)$ as before.
% All endpoints of segments crossing the chord of a bridge then have to be in the interior of the solution.
% With this variation, we do not need to consider the set $S_0$, we simply do not allow the bitangent $c$ to share an endpoint with a segment that crosses the current bridge.
% \alex{there might be subproblems for which no solution is possible.}
% 
% \rodrigo{Why don't we use this new initialization directly? Does it have any disadvantage over the current one?}

\begin{figure}
\centering
\includegraphics{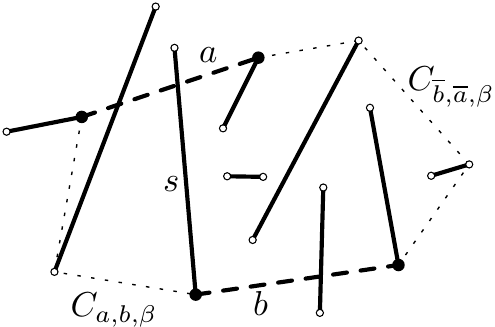}
\caption{A segment $s$ and two bitangents $a$ and~$b$ are chosen such that $s$ is the bridge connecting $a$ and~$b$.
Starting with this setting means that, in the maximization algorithm, the endpoint of a segment crossing the chord of a bridge cannot be a vertex of the maximum polygon.}
\label{fig_nc_maximize_initial}
\end{figure}

\section{Islands of points}\label{sec:islands}
A segment can be considered as the convex hull of two points.
From this point of view, we give a generalization of the algorithm for non-crossing segments to families of point sets whose convex hulls do not intersect pairwise.
The algorithm can be applied to the four variants of the problem (i.e. maximization/minimization of area/perimeter). 

Let $P$ be a set of $n$ points in the plane.
A \emph{cluster} is any subset of~$P$.
An \emph{island} $I \subset P$ is a cluster of $P$ such that $\CH(I) \cap P = I$;
see, e.g.,~\cite{bautista2011}. 
A pair of islands $(I_1, I_2)$ is called \emph{disjoint} if $\CH(I_1) \cap \CH(I_2) = \emptyset$.
Let $S_P$ be a set of islands partitioning a point set $P$.
Analogously to a set of segments, we say that an island $I \in S_P$ is \emph{stabbed} by a polygon $\mathcal{P}$ if one point of~$I$ is contained in $\mathcal{P}$, and $S_P$ is \emph{stabbed} by $\mathcal{P}$ if every island of $S_P$ is stabbed by $\mathcal{P}$.
In this section we show how to extend our algorithm for disjoint line segments to disjoint islands.

As in the previous section, consider a polygon $\mathcal{Q}$ spanned by $P$ and stabbing $S_P$.
Let $S_\mathrm{c}$ be the set of islands in~$S_P$ that intersect $\partial \mathcal{Q}$ \emph{at least} once.
Observe that, again, this definition of $S_c$ is slightly different from those 
considered for the previous problems.
%Note that this differs from the minimization case of line segments, just like in the proof of Theorem~\ref{thm:maximize} 
As shown in \figurename~\ref{fig_island_triangulation}, if an island is not a segment, $\partial \mathcal{Q}$ might intersect it several times and $\mathcal{Q}$ still contains a point of the island.
Let $X$ be the interior of $\mathcal{Q}$ and let $X' = X \setminus \bigcup_{I \in S_\mathrm{c}} \CH(I)$ (observe that this time the closure of $X'$ might be different from $\mathcal{Q}$, as the removed parts might have non-zero area).

The vertices of $X'$ are (i) vertices of the convex hulls of the islands that intersect with $\mathcal{Q}$, and (ii) the points where $\partial \mathcal{Q}$ crosses the convex hull boundary of islands in $S_P$. 
Note that the vertices of $\mathcal{Q}$ are a subset of~$P$, but they might not be on the convex hull of an island.
We say that $\mathcal{Q}$ \emph{crosses} an island~$I$ if an edge of~$\mathcal{Q}$ crosses an edge of~$\CH(I)$.
If an island contains only two points, we again consider $X'$ being ``slightly split'' at the 1-dimensional part corresponding to the convex hull of that island.
See \figurename~\ref{fig_island_triangulation}.

\begin{figure}
\centering
\includegraphics{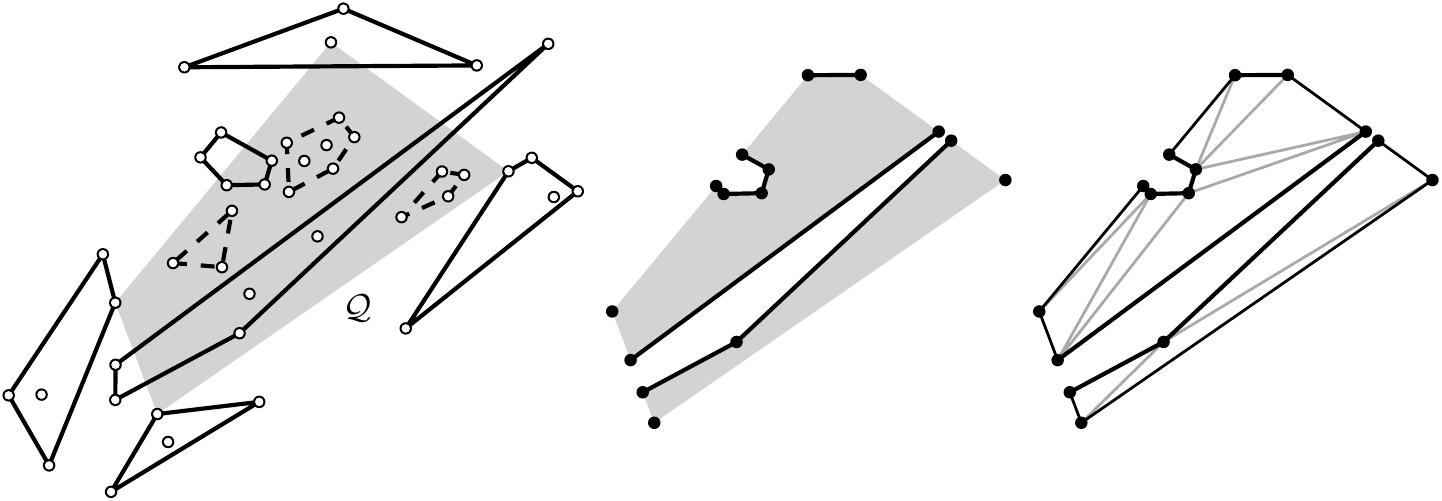}
\caption{Analogous to the algorithm for segments, there is a triangulation of the interior of a polygon with the convex hulls of the islands removed.}
\label{fig_island_triangulation}
\end{figure}

Let us first state a property of the maximization variant of the problem that also holds for general clusters.
The following is a generalization of Lemma~\ref{lem:LK} to clusters of points.

\begin{proposition}\label{prop_cluster_hull}
Given a set of clusters, there always exists a maximum (area or perimeter) stabbing polygon using only the points on the convex hull boundaries of the clusters.
\end{proposition}
\begin{proof}
Recall that, by Lemma~\ref{lem:LK}, when the points are chosen on line segments there is always a maximum stabbing polygon having the vertices on the endpoints of the segments~\cite{LK}.
Suppose there exists no maximum polygon with all vertices being extreme points of their clusters, and consider an optimal solution~$\mathcal{P}$, which has a point~$p$ in the interior of the convex hull of a cluster~$C$ as a vertex.
Pick any extreme point $q$ of $C$ (i.e., any vertex of $\CH(C)$) and let $s$ be the line segment that is defined by the intersection of the supporting line of $pq$ and $\CH(C)$.
Hence, one endpoint of $s$ is $q$ and the other endpoint, $\tilde{q}$, is on $\partial \CH(C)$, but is not an element of~$C$.
Let $S'$ be the set of line segments consisting of $s$ and one zero-length segment for each point in $\mathcal{P} \setminus \{q\}$. 
Applying Lemma~\ref{lem:LK} to $S'$ we conclude that there has to exist a larger polygon~$\mathcal{P}'$ containing an endpoint of $s$, and due to our assumption, this cannot be~$q$.
Thus, it has to be $\tilde{q}$. 
But $\tilde{q}$ is contained in an edge $s'$ of $\CH(C)$.
Again, we can define another set of line segments $S''$ that contains $s'$ and a zero-length segment for each point in $\mathcal{P}' \setminus \{ \tilde{q} \}$, and apply Lemma~\ref{lem:LK} to conclude that there has to exist a solution larger than $\mathcal{P}'$, and thus larger than $\mathcal{P}$, containing an endpoint of $s'$.
But any endpoint of $s'$ is an extreme point of $C$, a contradiction with the optimality of $\mathcal{P}$.
\end{proof}

Now consider again a set~$S_P$ of pairwise-disjoint islands.
Let $T$ be a triangulation of $X'$; $T$ again defines a set of chords that partitions $X'$.
An endpoint of a chord is either a vertex of the convex hull of some island, or the intersection of $\partial \mathcal{Q}$ with the convex hull boundary of some island.
% Each subproblem is defined by two islands and a chord that connects these two, as well as two edges of a stabbing polygon of which each intersects a different one of the two islands.
Given the set $B$ of all segments spanned by points of~$P$ that are part of different islands in $S_P$, and the crossings of segments in $B$ and edges of the convex hulls of the islands, we can apply the same algorithm as for segments:
Each subproblem is defined by two segments of $B$ that potentially define a stabber, and a bridge that is defined by either the convex hull of one island, or the two convex hulls of two islands and a chord (where the latter case also covers bridges where only one point of each convex hull is part of the bridge).
By the definition of $S_\mathrm{c}$, we assume that no island whose convex hull intersects the chord of a bridge intersects the boundary of the solution to the current subproblem.

\subsection{Structure of the subproblems}
When dealing with segments, the structure of a subproblem $(a,b,\beta)$ allowed to identify the chosen endpoints of the segments that formed $\beta$.
This aspect is more complicated when dealing with islands.

% Consider first the case where the bridge has a chord~$t$.
% By construction, $t$ is defined by two points (not necessarily in $P$) on $\partial \CH(I_1)$ and $\partial \CH(I_2)$, $I_1, I_2 \in S_P$.
Consider first the case where the bridge has a chord~$t$, and let~$I_1$ and~$I_2$ be the two islands that define~$t$ and whose convex hulls intersect~$a$ and~$b$, respectively.
The endpoints of~$t$ are on $\partial \CH(I_1)$ and $\partial \CH(I_2)$, but they are not necessarily points of~$P$.
% % By construction, we know that the chord is completely contained in the subproblem solution. Thus, if an endpoint of the chord is in~$P$ we are certain that the corresponding island already is in the solution (hence it can be ignored for the minimization version). The other case happens when one of the two islands that define $t$, say $\CH(I_1)$, is not a point of~$P$ (and therefore is at a crossing between $a$ and $\partial \CH(I_1)$. In this case, we know that the intersection of $\CH(I_1)$ with the solution is either completely inside or outside the subproblem region, due to convexity.
% % Otherwise, if the endpoint of $t$ on $\CH(I_1)$ is a point of $I_1$, we know that a point of $I_1$ is contained in $C^*_{a,b,\beta}$.
% % However, if we consider the maximization variant of our problem setting, we do not know whether we are allowed to pick a point of $I_1$ when constructing $C_{a,b,\beta}$;
Consider the endpoint of~$t$ at~$I_1$.
If this endpoint is also on~$a$ (because either an endpoint of $a$ is in~$I_1$ or $\partial \CH(I_1)$ intersects the interior of~$a$), then it is clear whether a point in~$I_1$ is picked inside the subproblem region or not, due to convexity ($C^*_{a,b,\beta}$ either enters or leaves $\partial \CH(I_1)$ at that endpoint of~$t$).
The analogous holds for $I_2$, see \figurename~\ref{fig_island_no_choice}.
Otherwise, if that endpoint of~$t$ is a vertex of $\CH(I_1)$ and not on~$a$, the algorithm has to make the decision of whether to pick a point of $I_1$ for $C_{a,b,\beta}$ or not.
For the minimization variant of the problem, the endpoint of $t$ in $I_1$ is always an optimal choice.
But for the maximization variant, the algorithm cannot locally decide whether it is better to have a point of $I_1$ on $C_{a,b,\beta}$ or on $C_{\overline{b}, \overline{a}, \beta}$ when solving the subproblem.
Again, the analogous holds for~$I_2$;
% 
% If neither endpoint of~$t$ is a point in~$P$, we know for each of $I_1$ and $I_2$ whether it is contained inside or outside of the subproblem region, due to convexity (the convex hull boundary of the island either enters or leaves $C^*_{a,b,\beta}$ at that endpoint of~$t$).
% Otherwise, suppose, w.l.o.g., that the endpoint of~$t$ at $I_1$ is actually a vertex of $\CH(I_1)$.
% We know, since $t$ is inside the subproblem solution by construction, that this endpoint is inside the subproblem solution, and therefore, there is no reasonable choice of another point of~$I_1$ in the minimization variant of the algorithm.
% However, for the maximization variant, the algorithm may pick a point of~$I_1$ to be part of $C_{a,b,\beta}$ or not;
see \figurename~\ref{fig_island_unknown_bridge}.
We solve this in the following way.
Recall that, due to Proposition~\ref{prop_cluster_hull}, we only need to consider the vertices of the convex hull of~$I_1$ for the maximization variant.
When the algorithm has to divide a subproblem into two further subproblems, and has to pick a point of an island~$I_1$, it passes a parameter to one of the two subproblems indicating that a point of~$I_1$ being incident to the subproblem's region has to be picked, and afterwards tests the same subproblem combination, this time indicating that the point of $I_1$ belongs to the other subproblem.

The case where the bridge consists of only one island~$I$ (and hence does not contain a chord; see \figurename~\ref{fig_island_no_chord}) is similar. 
In this case, the same issue arises for both the maximization and minimization variant;
we do not know whether the selected point of~$I$ has to be inside the subproblem region or not.
However, this can also be indicated to the subproblem with a single flag.

Summing up, a subproblem is defined by the following elements:
\begin{itemize}
 \item the two segments $a$ and $b$, and
 \item the bridge~$\beta$, consisting of\begin{itemize}
  \item two islands $I_1$ and $I_2$ or a single island~$I$,
  \item a chord~$t$ (if there are two islands in the bridge), and
  \item a flag for each given island. Each flag indicates whether a point of $I_1$ or $I_2$ (or~$I$) is picked for the solution of the subproblem or not.
 \end{itemize}
\end{itemize}
Even though the subproblem definition became more complex by the generalization, the number of subproblems is still polynomial in~$|P|$:
Conceptually, a subproblem can be guessed in the following way:
We pick two pairs of points from $P$ representing $a$ and $b$.
For $a$, we pick either another point~$p$ from $P$, contained in some island $I_1$ intersecting $a$, or one of the $O(n)$ edges on the convex hull of an island $I_1$.
In the first case, we suppose that $p$ is the endpoint of the chord~$t$, in the second case $t$ ends in the intersection of $a$ with the guessed convex hull edge of~$I_1$.
The same is done for~$b$.
The case where the bridge consists of a single island~$I$ requires only guessing $a$, $b$, and $I$.
Hence, there are $O(|P|^6)$ subproblems, as in the previous section (where all islands had two elements).

% However, if the bridge does not have a chord, we do not know for a subproblem whether we have to include a point of the islands that form the bridge or not (if there is a chord, the endpoints of the chord are picked anyway if they are part of~$P$, or otherwise the convex hull of the corresponding island has its intersection with the solution either entirely inside or outside the subproblem region);
% see \figurename~\ref{fig_island_unknown_bridge}.
% However, this can easily be figured out from the settings of which such a subproblem is part of.
% We therefore calculate two solutions for such subproblems, one that has to contain a point of the bridge island and one that does not.\mati{This paragraph should have many more details. If I understood correctly it is the main difference}

\begin{figure}
\centering
\includegraphics{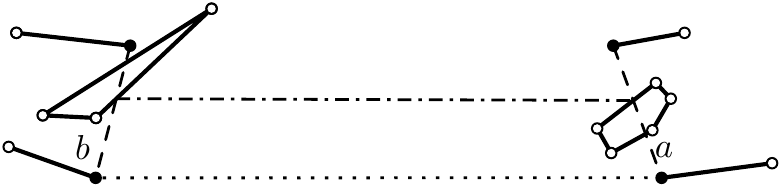}
\caption{An example of a case where the bridge determines whether a point of an island defining the bridge needs to be picked inside the subproblem region or not.}
\label{fig_island_no_choice}
\end{figure}

\begin{figure}
\centering
\includegraphics{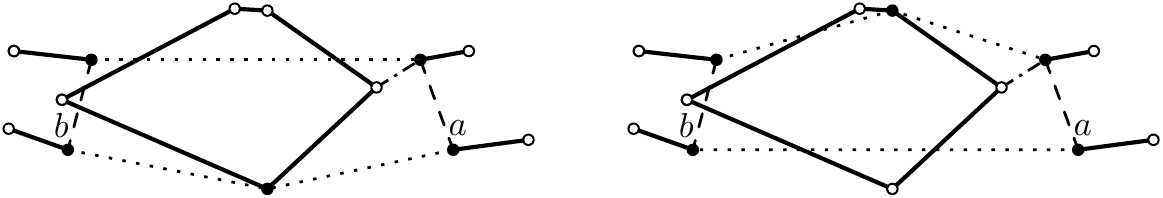}
\caption{For every island contained in the bridge, we actually have two cases: one where a point of the island has to be contained in the subproblem solution, and one where it must not contribute a vertex. This cannot be determined locally, so we consider both cases.}
\label{fig_island_unknown_bridge}
\end{figure}

\subsection{Choice of the subproblems}
Choosing the subproblems is also more sophisticated for islands than for line segments.
Again, we want to choose a triangle $\Delta_d$ for each subproblem $(a,b,\beta)$ by guessing the apex of the triangle and the edge $c$ defining it.
When dealing only with segments, $\Delta_d$ was attached to the edge $t$ of the bridge~$\beta$.
If the bridge contains a chord $t$, the cases are the same: either $t$ is part of $C_{a,b,\beta}$, or, for some choice of $d$, $\Delta_d$ is a triangle in a triangulation of $C_{a,b,\beta}$.
Unlike when dealing only with segments, we must also consider the case where $\beta$ consists of a part of the convex hull boundary of a single island~$I$.
(It may even occur that $\beta$ is intersected by $C_{a,b,\beta}$ at another bitangent $c$ of $C_{a,b,\beta}$;
note that $c$ might or might not have an endpoint in~$I$.)
Since $\beta$ is not a single edge, we need a different well-defined way to choose the edge of triangle~$\Delta_d$ not containing~$d$.
% We therefore first check all possible bitangents $c$ that intersect $\beta$, these give us subproblems whose bridges are defined solely by parts of $\beta$.
% Then, we check all triangles that share an edge with~$\beta$.
Observe that both points $a_2$ and $b_2$ cannot be inside $\CH(I)$, as this would mean that two points of $I$ are chosen for the boundary of the solution.
W.l.o.g., let $a_2$ be outside $\CH(I)$.
Let $a_\beta$ be the point on $a$ intersecting $\partial\CH(I)$ that is closer to~$a_2$.
Then there is a part of $a$ starting at $a_\beta$ towards $a_2$ that is not contained inside the convex hull of any island.
If $a$ is an edge of the optimal polygon, this part of $a$ is an edge of $X'$.
Therefore, we can choose it as the base edge $t$ of $\Delta_d$, and $\Delta_d$ is part of a triangulation of the optimal polygon for some choice of~$d$.
See again \figurename~\ref{fig_island_no_chord}.

\begin{figure}
\centering
\includegraphics{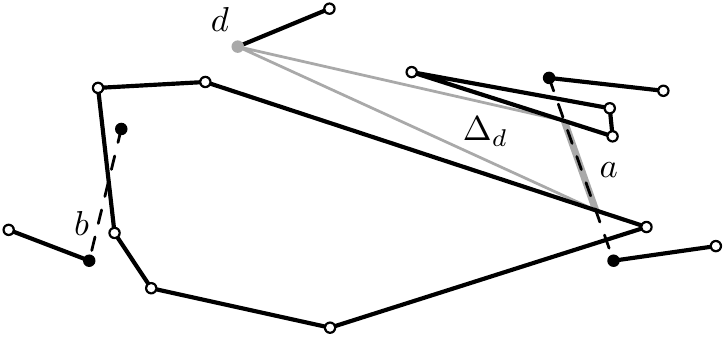}
\caption{If the bridge is defined by the convex hull of a single island, the base edge of the triangle~$\Delta_d$ (shown fat gray) is chosen w.l.o.g.\ on~$a$. Observe that also other convex hulls of islands might intersect the defining edge, but no such convex hulls are part of $X'$ if $a$ is part of the optimal polygon.}
\label{fig_island_no_chord}
\end{figure}

\subsection{Initialization}
There is another technical difficulty related to the segments that intersect the first bridge:
For these, we need a way to determine whether there is already a part on the inside, and in addition, there are more cases to consider than for segments in the proof of Theorem~\ref{thm:maximize}.
However, we can apply the same trick as discussed in Section~\ref{sec:maxarea}.
When guessing the initial subproblem, we do not guess three consecutive edges, but two ``opposite'' edges for which there exists an island $I$ whose convex hull is intersected by both edges.
The two edges and $I$ define two subproblems.
After checking all such pairs of edges, we know that all remaining solution candidates are in our classical setting, i.e., if an island intersects the chord of a bridge, it has to be inside the partial solution.
Note that in this special case it may occur that an island intersects both the chord~$\beta_0$ and one of the edges that define the initial subproblem;
in any case, we know that the remaining part of the island has to be inside the resulting polygon.

Hence, the analogous statements to the lemmata in the previous section hold, and we have a polynomial-time dynamic programming algorithm that finds an optimal solution for any of the four variants of the problem.%
\footnote{Concerning the running time, the same observations hold as in the previous section.
We have $O(|P|^6)$ subproblems.
In each subproblem, we basically have to choose the point $d$ and the bitangent intersecting the island containing $d$.
After this guess, all necessary conditions can be checked in linear time.}

\begin{theorem}\label{thm_islands}
Given a point set $P$ in the plane and a partition $S_P$ of $P$ such that the convex hulls of any two elements of $S_P$ are disjoint, we can solve any of the four variants of the {\sc SPP} in polynomial time.
\end{theorem}

Observe that convexity of the stabbed sets is crucial for our approach. Schlipf~\cite{schlipf} showed that finding a convex transversal is NP-complete if the stabbed sets are non-convex, even if they are disjoint.
Her reduction can be adapted to our setting in a way similar to the one shown in the next section.

\section{NP-hardness of the {\sc MinPerSPP} and the {\sc MinAreaSPP}}\label{sec:hardness}

In this section we prove that the {\sc MinPerSPP} and the {\sc MinAreaSPP} are NP-hard by a reduction from 3-SAT.
Note that the NP-hardness of the maximization variants of the {\sc SPP} are already known~\cite{ich-waac-11,LK}.

Our reduction has a structure very similar those used in~\cite{ADKMPSY,ich-waac-11,daescu2010,LK}. In the following we give the construction, not only for completeness, but also because it will be used afterwards to show hardness of non-crossing clusters. In addition, we also provide a full proof that the coordinates of the points can be described in polynomial time. This is mentioned in the previous constructions, but details are often omitted.

%Moreover,  , although the details are different, and we include a full proof of the fact that the coordinates of the points used can be described in polynomial time.
%\rodrigo{Added the previous sentence, perhaps someone can improve it a bit}
%
%}.
%
%
% In the \pb{MAX-E3-SAT} the input consists of $n$ boolean variables
% $x_1,\ldots,x_n$ and $m$ clauses $C_1,\ldots,C_m$, over $x_1,\ldots,x_n$
% with exactly three literals each, and the
% output is an assignment of values to $x_1,\ldots,x_n$ so that the number of
% satisfied clauses  is maximized~\cite{hastad}.

\begin{theorem}\label{thm_np_hard}
%Given a set $S$ of segments and a number~$k$, it is NP-hard to decide whether there exists a stabbing polygon for $S$ with perimeter of length at most~$k$.
The {\sc MinPerSPP} and the {\sc MinAreaSPP} are NP-hard.
\end{theorem}

\begin{proof}
For ease of exposition, we present the proof for the {\sc MinPerSPP}.
As mentioned in the end of the section, the adaptations needed for the construction to work for the {\sc MinAreaSPP} are minor.

Let a 3-SAT instance consist of $n$ variables
$x_1,\ldots,x_n$ and $m$ clauses $C_1,\ldots,C_m$.
We reduce this instance to the following one of
the {\sc MinPerSPP}. We draw a circle and place variable gadgets in the left semicircle,
clause gadgets in the right semicircle, and segment connectors joining variable
gadgets with clause gadgets. See \figurename~\ref{fig:np-reduction}a. 
%Given any two points $p$ and $q$ in the plane, let $pq$ denote the straight segment
%joining $p$ and $q$, and $|pq|$ denote the length of $pq$.

\begin{figure}[bt]
% 	\centering
% 	\subfloat[]{
% 		\includegraphics[width=3cm]{np-proof-overview}
% 		\label{fig:np-overview-2}
% 	}
% 	\hspace{1cm}
% 	\subfloat[]{
% 		\includegraphics[scale=0.4]{np-proof-variable-gadget}
% 		\label{fig:np-variable-gadget-2}
% 	}
% 	\hspace{1.5cm}
% 	\subfloat[]{
% 		\includegraphics[scale=0.4]{np-proof-clause-gadget}
% 		\label{fig:np-clause-gadget-2}
% 	}
% 	\label{fig:np-reduction-2}
\centering
\includegraphics{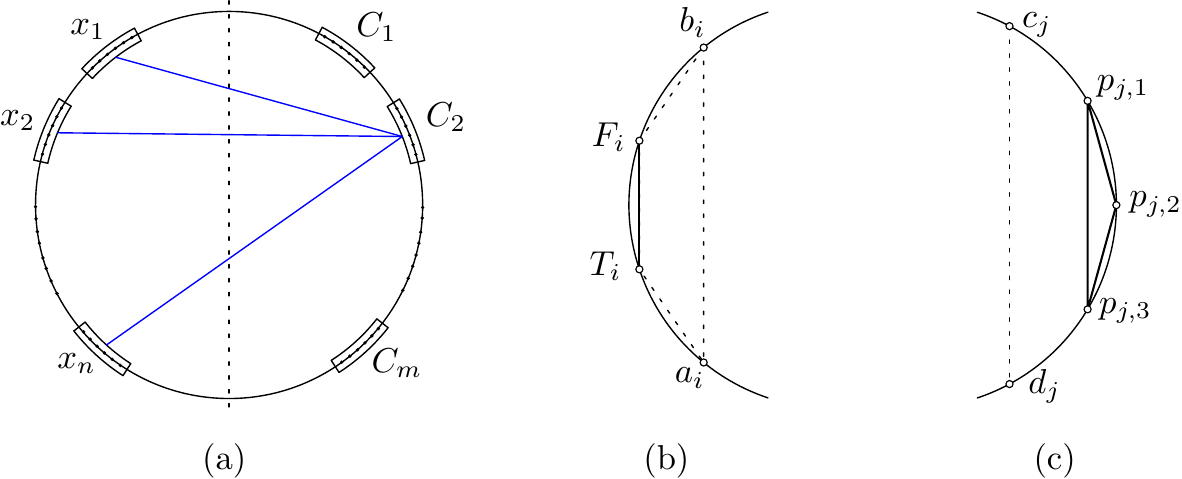}
		\caption{(a) Overview of the reduction from the \pb{3-SAT}.
		Variable gadgets (b) are to the left and clause gadgets (c) to the right. 
	}
\label{fig:np-reduction}
\end{figure}

% \rodrigo{We should unify the ``subfigure'' style, and also how we refer to them: 7(b) vs 7b}
% \alex{I did not manage to do so...}

%\paragraph*{Variable Gadgets}
%\paragraph*{Gadgets}
\paragraph{Gadgets}
For each variable $x_i$, $i\in[1..n]$,
we put points $T_i$ and $F_i$ on the circle and place three segments:
segment $T_iF_i$, and two zero-length segments $a_i$ and $b_i$,
so that $T_iF_i$ is parallel to the line containing both $a_i$ and $b_i$.
Refer to \figurename~\ref{fig:np-reduction}b.
Furthermore, trapezoids with vertices $a_i,T_i,F_i,b_i$, for all $i\in[1..n]$, are congruent.
Let $P_v:=|a_iT_i|+|T_ib_i|=|a_iF_i|+|F_ib_i|$. 
%and $P'_v:=|a_iT_i|+|T_iF_i|+|F_ib_i|$.

%\paragraph*{Clause gadgets}
For each clause $C_j$, $j\in[1..m]$,
we first place two zero-length segments $c_j$ and $d_j$.
We select three points $p_{j,1}$, $p_{j,2}$, and $p_{j,3}$,
dividing evenly the smallest arc of the circle joining $c_j$ and $d_j$ into four arcs,
and then we place three other segments: $p_{j,1}p_{j,2}$,
$p_{j,2}p_{j,3}$, and $p_{j,3}p_{j,1}$.
See \figurename~\ref{fig:np-reduction}c. The convex pentagons with
vertices $d_j,c_j,p_{j,1},p_{j,2},p_{j,3}$, for all $j\in[1..m]$, are congruent.
Let $P_c:=|c_jp_{j,1}|+|p_{j,1}p_{j,2}|+|p_{j,2}d_j|=|c_jp_{j,1}|+|p_{j,1}p_{j,3}|+|p_{j,3}d_j|=
|c_jp_{j,2}|+|p_{j,2}p_{j,3}|+|p_{j,3}d_j|$.
% and $P'_c:=|c_jp_{j,1}|+|p_{j,1}p_{j,2}|+
%|p_{j,2}p_{j,3}|+|p_{j,3}d_j|$.
%We further ensure that $m(P'_c-P_c)<P'_v-P_v$. This condition will be necessary in the problem
%reduction.

%\paragraph*{Connectors}
For each clause $C_j$, $j\in[1..m]$, we add segments called \emph{connectors} as follows.
Let $x_i$ be the variable involved in the first literal of $C_j$. If $x_i$ appears in positive form
then let $\overline{p_{j,1}}$ denote the point $T_i$. Otherwise, if $x_i$ appears in negative form,
then let $\overline{p_{j,1}}$ denote the point $F_i$.
In both cases we add the connector $\overline{p_{j,1}}p_{j,1}$. 
We proceed analogously with the variable in the
second literal and point $p_{j,2}$, and with the variable in the third literal and point $p_{j,3}$.

\newcommand{\minper}{\ensuremath{m}}

%\paragraph*{Problem reduction} 
\paragraph{Problem reduction}
Consider the instance of the {\sc MinPerSPP} consisting of the set of segments added at variable gadgets,
clause gadgets, and connectors.
Observe that any optimal polygon $\mathcal{P}_\mathrm{opt}$ for this instance satisfies the following
conditions:
\begin{itemize}
\item[(a)] For each variable $x_i$, $i\in[1..n]$, $\mathcal{P}_\mathrm{opt}$
contains as vertices the points $a_i$ and $b_i$, and at least
one point between $T_i$ and $F_i$.

\item[(b)] For each clause $C_j$, $j\in[1..m]$, $\mathcal{P}_\mathrm{opt}$
contains the points $c_j$ and $d_j$ as vertices, and at least
two points among $p_{j,1}$, $p_{j,2}$, and $p_{j,3}$.

\item[(c)] For each clause $C_j$, $j\in[1..m]$, $\mathcal{P}_\mathrm{opt}$
contains exactly two points among $p_{j,1}$, $p_{j,2}$, and $p_{j,3}$  as vertices
if and only if it contains at least one point among 
$\overline{p_{j,1}}$, $\overline{p_{j,2}}$, and $\overline{p_{j,3}}$ as vertex,
located in variable gadgets. See \figurename~\ref{fig:np-transversing-clauses}.

%$\mathcal{P}_\mathrm{opt}$ contains as vertices
%points $a_i$ and $b_i$ for all variables $x_i$, $i\in[1..n]$, and
%points $c_j$ and $d_j$ for all clauses $C_j$, $j\in[1..m]$.
%
%\item[(b)] For each variable $x_i$, $i\in[1..n]$,
%$\mathcal{P}_\mathrm{opt}$ contains exactly one of $T_i$ and $F_i$ as vertex between $a_i$ and $b_i$.

\item[(d)] The perimeter of $\mathcal{P}_\mathrm{opt}$ is at least
the minimum possible perimeter
\[
\minper=|b_1c_1|+|a_nd_m|+\sum_{i=1}^{n-1}|a_ib_{i+1}|+\sum_{j=1}^{m-1}|d_jc_{j+1}|+nP_v+mP_c \enspace .
\]

If $\mathcal{P}_\mathrm{opt}$ has perimeter $\minper$, $\mathcal{P}_\mathrm{opt}$ contains, for every $i\in[1..n]$, exactly one point between $T_i$ and $F_i$ as vertex, 
and contains, for every $j\in[1..m]$, exactly two points among $p_{j,1}$, $p_{j,2}$, and $p_{j,3}$ as vertices.
\end{itemize}
{
Condition (a) follows from the fact that $a_i$ and $b_i$, $i\in[1..n]$, are zero-length segments,
where at least one endpoint from each of them must be contained in $\mathcal{P}_\mathrm{opt}$, and the
presence of the segment $T_iF_i$.
Condition (b) is due to the fact that $c_j$ and $d_j$, $j\in[1..m]$, are zero-length segments
and the presence of the segments $p_{j,1}p_{j,2}$, $p_{j,2}p_{j,3}$, and $p_{j,3}p_{j,1}$.
Condition (c) follows from the presence of the segments $p_{j,1}p_{j,2}$, $p_{j,2}p_{j,3}$, and $p_{j,3}p_{j,1}$,
together with the connectors $\overline{p_{j,1}}p_{j,1}$, $\overline{p_{j,2}}p_{j,2}$, and $\overline{p_{j,3}}p_{j,3}$.
Condition (d) follows from (a) and (b).

\begin{figure}[bt]
	\centering
	\includegraphics{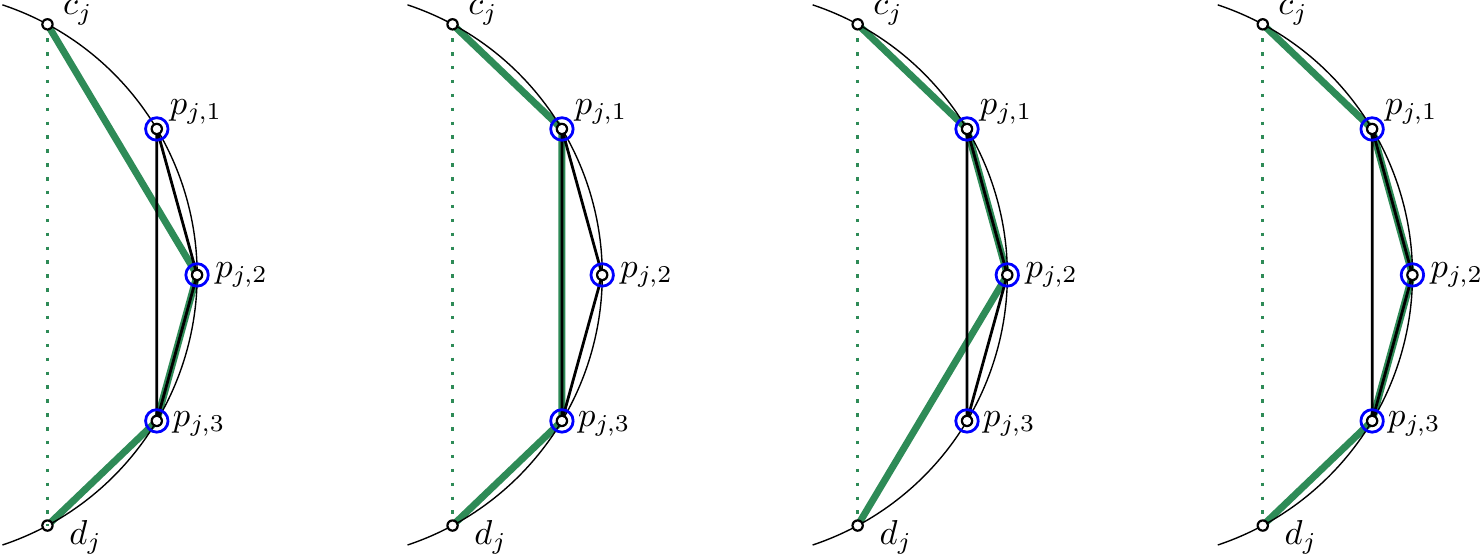}
	\caption{
		In each clause gadget, if at least one connector has an endpoint
		not in the gadget as a vertex of optimal polygon $\mathcal{P}_\mathrm{opt}$,
		then exactly two points among $p_{j,1}$, $p_{j,2}$, and $p_{j,3}$
		are vertices of $\mathcal{P}_\mathrm{opt}$ (any of the first three figures from left to right).
		Otherwise, all three are vertices of $\mathcal{P}_\mathrm{opt}$ (rightmost figure).
	}
	\label{fig:np-transversing-clauses}
\end{figure}

Let $\mathcal{P}$ be any feasible polygon satisfying conditions (a)-(c) and having minimum perimeter~$\minper$.
Polygon $\mathcal{P}$ induces the following assignment for the variables
$x_1,x_2,\ldots,x_n$: we assign $x_i$ to true if point
$T_i$ is a vertex of $\mathcal{P}$, false otherwise.
With this assignment we can ensure that every clause $C_j$ is satisfied if and only if
exactly two points among $p_{j,1}$, $p_{j,2}$, and $p_{j,3}$ are vertices of
$\mathcal{P}$.
% Let $k$ be the number of clauses satisfied by this assignment. Let $\pi(\mathcal{P})$ denote the perimeter of $\mathcal{P}$ and $P_0:=|b_1c_1|+|a_nd_m|+\sum_{i=1}^{n-1}|a_ib_{i+1}|+\sum_{j=1}^{m-1}|d_jc_{j+1}|$,
% then:
% \begin{eqnarray}
% \pi(\mathcal{P})&=&P_0+nP_v+kP_c+(m-k)P'_c= P_0+nP_v+mP'_c-k(P'_c-P_c) \enspace 
% \label{eqn2}
% \end{eqnarray}
% 
% Conversely, each assignment for $x_1,x_2,\ldots,x_n$ satisfying exactly $k$ clauses,
% induces a feasible polygon $\mathcal{P}$ satisfying 
% conditions (a)-(c) and equation (\ref{eqn2}). 
Therefore, the 3-SAT formula consisting of the clauses $C_1,\ldots,C_m$ is satisfiable if and only if the perimeter of the {\sc MinPerSPP} is~$\minper$.

% The points of the above gadgets can be placed as $\Theta(nm)$ rational points on the unit circle.
%, with the angle between every two consecutive points of a variable gadget being $\Omega(m)$ times the angle between every two points of a clause gadget, to ensure $m(P'_c-P_c)<P'_v-P_v$.
To complete the proof, it remains to show that the coordinates specifying the positions of the points of the gadgets can be expressed as rational numbers whose size is polynomial in $n$ and $m$.
The proof of this fact is relegated to the appendix.

Note that our reduction uses segments of zero length. 
This can be avoided by 
%locating  an endpoint of each segment sufficiently far 
%from the circle where the gadgets are placed on, in order to guarantee that no optimal solution can contain any of these endpoints.
%
replacing each zero-length segment by a sufficiently short segment, in order to guarantee that the choice of the endpoint makes only a marginal difference in the overall cost of any solution.
}\end{proof}

%\mati{The text from here downwards should be in the conclusion. What do you think?}
Observe that the same reduction with minor modifications
applies for the case of minimizing the
area of the output polygon, i.e., for the {\sc MinAreaSPP}.
Moreover, our proof shows that the problem
remains NP-hard even if the endpoints of all the segments
are in convex position or lie on a circle.
Recently, it has been shown that the case in which the segments are diameters of a circle, both the minimum and maximum area problems can be solved in linear time~\cite{diameter}.

\subsection{Fixed-parameter tractability}
It is worth mentioning that the four variants of the {\sc SPP} are fixed-parameter tractable (FPT) on the number $k$ of segments that 
intersect other segments. 
Namely, let $S'\subseteq S$ be the set of segments of $S$ that do not
intersect any segment of $S$. 
Consider the $2^k$ instances of  {\sc SPP} such that each consists 
of the elements of $S'$ joint with exactly one endpoint (i.e., a segment of length zero) 
of each element of $S\setminus S'$. All these instances can be solved in 
$O(2^k P(n))$ time, for the polynomial
time $P(n)$ of Theorem~\ref{teo:main} or Theorem~\ref{thm:maximize}, since each instance
consists of pairwise disjoint segments. 
The optimal solution for $S$ is among the $O(2^k)$ solutions found for those instances. 
We summarize with the following observation.

\begin{observation}\label{obs_fpt}
Given a set $S$ of $n$ segments, any of the four variants of the {\sc SPP} can be solved in $O(2^k P(n))$ time, where $k$ is the number of segments in $S$ that cross at least another segment from $S$, and $P(n)$ is the running time of the algorithm from Theorem~\ref{teo:main} or Theorem~\ref{thm:maximize} (depending on the variant).
\end{observation}

\subsection{Generalization to non-crossing clusters}
In Section~\ref{sec:islands}, we generalized the algorithm from segments to islands, i.e., subsets of a point set~$P$ whose convex hulls do not intersect.
Next we remark that relaxing the disjointness only slightly, allowing \emph{non-crossing} clusters, results again in an NP-hard problem.

Recall that a \emph{cluster} is any subset of~$P$.
Two clusters~$J_1$ and $J_2$ are \emph{non-crossing} if the convex hull boundary of $J_1 \cup J_2$ has at most two segments with one endpoint in~$J_1$ and the other endpoint in~$J_2$~\cite{clusters}.
In these terms, the reduction presented uses crossing clusters of size 2.

It is conceivable that 
the key factor in the NP-hardness of the problem relies in allowing the gadgets to cross, since the difference between crossing and non-crossing clusters plays a role in the time complexity of other problems, e.g.~\cite{pl-hvdpo-04}. 
Thus, it could be the case that if clusters can intersect but not cross the problem becomes polynomial-time solvable. 
However, we show next that our reduction can be adapted to non-crossing clusters as well.

To every segment (i.e., cluster of size 2) that connects a variable gadget to a clause gadget, we add another two points that are ``far away'' from the circle in which we place our gadgets, such that the resulting clusters are non-crossing and the points not part of the gadget never get chosen.
Suppose w.l.o.g.\ that the variable gadgets are in the third quadrant of the plane and the clause gadgets are in the first quadrant.
Consider the two lines $\ell_\mathrm{c} : y = 2x + d$ and $\ell_\mathrm{v} : y = (x+d)/2$, for some constant~$d>0$ to be made more precise later. 
For every segment $e_\mathrm{v} e_\mathrm{c}$ that connects a variable gadget with a clause gadget, we project the point $e_\mathrm{v}$ horizontally on $\ell_\mathrm{v}$ and the point $e_\mathrm{c}$ vertically on $\ell_\mathrm{c}$, and add the two projected points to the cluster of the segment; see \figurename~\ref{fig_clusters_reduction}.
Regarding the value of $d$, it suffices to choose a value large enough such that the circle of the construction is below $\ell_\mathrm{c}$ and $\ell_\mathrm{v}$, and far enough from the lines so that the projected points can never be part of an optimal solution.

\begin{figure}
\centering
\includegraphics{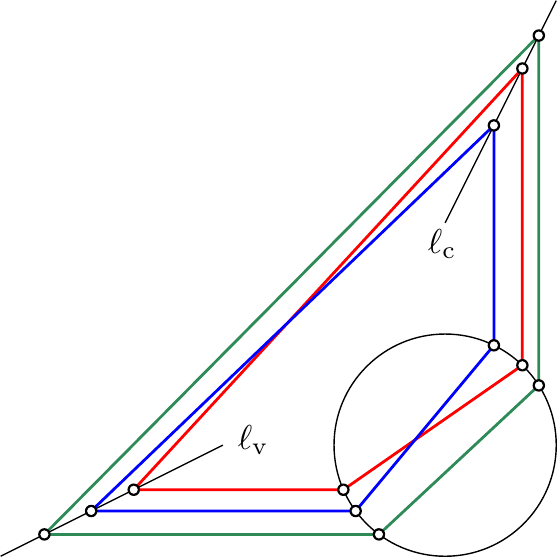}
\caption{Extending segments to clusters of size 4 in convex position.
Each pair of the resulting clusters is non-crossing.}
\label{fig_clusters_reduction}
\end{figure}

Observe that the resulting clusters are in convex position and non-crossing.
Further, no such cluster crosses a segment that is part of a gadget.
None of the new points can be part of the optimal solution if the constant~$d$ is chosen sufficiently large, and hence, the construction behaves in the same way as the construction for segments.
The construction can be easily altered to give a point set in general position by replacing the relevant segments on $\ell_\mathrm{v}$ and $\ell_\mathrm{c}$ by, say, sufficiently flat circular arcs.

\begin{theorem}
Given a set of non-crossing clusters, it is NP-hard to find an optimal solution of the {\sc MinPerSPP} or the {\sc MinAreaSPP}.
The problems remain NP-hard if each cluster is in convex position.
\end{theorem}

While the structure of the reduction for line segments is similar to the one for the maximization variant by L\"offler and van Kreveld~\cite{LK}, the adaption for non-crossing clusters cannot be done in the same way.
However, for the maximization variant, Proposition~\ref{prop_cluster_hull} applies.

Therefore, for solving the following problem, it is sufficient to consider the cases where each cluster is in convex position.

\begin{open}
What is the complexity of finding a maximum (area or perimeter) stabbing polygon of non-crossing clusters of points?
That is, find a maximum stabbing polygon whose vertices are elements of~$P$ such that no two belong to the same cluster.
\end{open}

\bibliographystyle{abbrv}
\bibliography{refs}

%\mati{NOTE: We removed lots of sections from before. They can all be found in the extras.tex document}
%\input{extras.tex}

\newpage
\appendix

\section{Exact point construction for Theorem~\ref{thm_np_hard}}\label{sec:coordinates}
In this section we give the exact construction for the segment set described in the proof of Theorem~\ref{thm_np_hard}.
We place the segment endpoints of the gadgets at rational coordinates on the unit circle in such a way that the size of both the numerator and denominator of each coordinate are bounded by a polynomial function of the size of the initial problem.

%Canny, Donald, and Ressler~\cite{canny} give an account on how to efficiently select rational points (i.e., points with rational coordinates) on the unit circle.
With respect to the segment endpoints, both the variable gadgets and the clause gadgets have the same structure.
For any point $p$, let $\angle p$ denote the polar angle of the point.
If we construct a variable gadget such that $\alpha_v := \angle a_i - \angle T_i = \angle T_i - \angle F_i = \angle F_i - \angle b_i$, the gadget works as described.
If we use the same angle $\alpha_v$ for all variable gadgets, the gadgets are congruent.
The analogous holds if we choose an angle $\alpha_c = \angle c_j - \angle p_{j,1} = \ldots = \angle p_{j,3}-\angle d_j$ between two consecutive points when constructing a clause gadget; we obtain a set of congruent gadgets that fulfill the properties described previously.
In the remainder of this section we show how to choose the reference points $a_i$ and $c_j$ for each gadget among the rational points on the unit circle and the angles $\alpha_v$ and $\alpha_c$.

We use several well-known facts about rational points on the unit circle (see, e.g., \cite{husemoeller}).
For any $t\in \mathbb{Q}$, the point $p_t = \left (\frac{t^2-1}{t^2+1}, \frac{2t}{t^2+1} \right )$ is rational and lies on the unit circle.
Hence, the coordinates of $p_t$ describe the cosine and sine, respectively, of the polar angle $\angle p_t$.
%We will restrict ourselves to the case $t\in\mathbb{N}$.
Observe that for any $t>1$, point $p_t$ lies in the first quadrant.\footnote{Note that there are multiple conventions for choosing the sign of the coordinates in this parametrization.}
In particular, we have $\angle p_t \rightarrow 0$ and $p_t \rightarrow (1,0)$ when $t \rightarrow +\infty$.
Using the trigonometric identity $
\sin(\alpha_1 - \alpha_2) = \sin( \alpha_1) \cos( \alpha_2) - \cos( \alpha_1) \sin( \alpha_2)$
we obtain:
\begin{align*}
\sin(\angle p_{t} - \angle p_{t+1}) &=
\left( \frac{2t}{t^2 + 1} \right)\left( \frac{(t+1)^2-1}{(t+1)^2+1}\right) - \left( \frac{t^2-1}{t^2+1}\right) \left( \frac{2(t+1)}{(t+1)^2+1}\right)\\
&= \frac{2(t^2+t+1)}{(t^2+1)(t^2+2t+2)} \enspace .
\end{align*}

\begin{observation}\label{obs_angle}
For $t \geq 1$, the angle $\angle p_t - \angle p_{t+1}$ is monotonically decreasing in $t$.
\end{observation}
The following inequality will allow us to choose both the reference points and the small angles for the gadget construction.
If $1 \leq t < 5N$ for some positive integer $N$ (the factor 5 is chosen with foresight), we have

%\begin{align*}
%\sin(\angle p_{t} - \angle p_{t+1}) 
%&\geq \frac{2(25N^2+5N+1)}{(25N^2+1)(25N^2+10N+2)}
%\geq \frac{62}{962N^2} %\enspace .
%\end{align*}

\begin{align*}
\sin(\angle p_{t} - \angle p_{t+1}) 
\geq \frac{2(25N^2+5N+1)}{(25N^2+1)(25N^2+10N+2)}\\
\geq \frac{50N^2}{(25N^2+1)(25N^2+10N+2)}
=  \frac{50N^2}{625N^4+250N^3+75N^2+10N+2}\\
\geq  \frac{50N^2}{625N^4+250N^4+75N^4+10N^4+2N^4}
= \frac{50}{962N^2} \enspace .
\end{align*}

We can use this bound to define a rational angle $\angle p_s$ that is smaller than all intervals we consider in the construction:
\[
\sin(\angle p_s) = \frac{2s}{s^2+1} < \frac{2s}{s^2} \stackrel{!}{\leq} \frac{50}{962N^2} \enspace ,
\]
which is fulfilled if
\[
s \geq \frac{962N^2}{25} \enspace .
\]

Let us place the endpoints for the variable gadgets.
We choose $a_i = p_{(5i-4)}$ and therefore set $N = n$.
Further, we choose $s = 100n^2$, which fulfills the above inequalities.
By the choice of $s$, we have $\angle a_{i+1}+3\angle p_s <\angle a_{i}$, hence the variable gadgets do not interfere with each other.
We therefore can place $T_i, F_i,$ and $b_i$ on the arc segment between $a_i$ and $a_{i+1}$ by rotating $a_i$ up to three times by $\alpha_v = \angle p_s$.
The points can be explicitly computed from $a_i$ by using the coordinates of $p_s$ as elements of the rotation matrix:
\[
 \begin{pmatrix} \cos(\angle p_s) & \sin(\angle p_s)\\ -\sin(\angle p_s) & \cos(\angle p_s)  \end{pmatrix}^{\!\!\kappa} \cdot \begin{pmatrix} \cos(\angle a_i)\\ \sin(\angle a_i) \end{pmatrix}
\]
for $\kappa \in \{1, 2, 3\}$.
Observe that the coordinates of $a_i$ and $p_s$ are the sines and cosines of the corresponding angles and are rational; therefore, the resulting points also have rational coordinates, which are bounded by a polynomial function in the size of the input, since $\kappa$ is constant.
Finally, we change the signs of the coordinates, so that the variable gadgets are placed on the third quadrant.

For the clause gadgets, we can basically proceed in the same manner, choosing $N = m$ in the above equation, as well as $c_j = p_{(5j-4)}$.

\end{document}